\documentclass[11pt]{llncs}

\usepackage{amsmath}
\usepackage{amsfonts}
\usepackage{graphicx}
\usepackage{float}
\usepackage{subfig}
\usepackage{algorithm}
\usepackage[noend]{algpseudocode}
\usepackage[colorinlistoftodos]{todonotes}
\usepackage{color}
\usepackage{paralist}

\newcommand{\triplet}{R} 

\algdef{SE}[SUBALG]{Indent}{EndIndent}{}{\algorithmicend\ }%
\algtext*{Indent}
\algtext*{EndIndent}

\title{Constructing a Consensus Phylogeny from a Leaf-Removal Distance}
\author{Cedric Chauve\inst{1}, Mark Jones\inst{2}, Manuel Lafond\inst{3}, C\'eline Scornavacca \inst{4}, Mathias Weller\inst{5}}

\institute{
  Department of Mathematics, Simon Fraser University, Burnaby, Canada. 
  \email{cedric.chauve@sfu.ca}
  \and
  Delft Institute of Applied Mathematics,
  Delft University of Technology,
  P.O. Box 5, 2600 AA, Delft, the Netherlands.
 \email{M.E.L.Jones@tudelft.nl}
    \and Department of Mathematics and Statistics, University of Ottawa, Ottawa, Canada.  
  \email{mlafond2@uOttawa.ca}
  \and
  Institut des Sciences de l'Evolution, 
 Universit\'{e} de Montpellier, CNRS, IRD, EPHE, Montpellier - France.
 \email{Celine.Scornavacca@umontpellier.fr}
 \and 
  Laboratoire d'Informatique, de Robotique et de
 Micro\'{e}lectronique de Montpellier, IBC, Montpellier - France.
 \email{mathias.weller@lirmm.fr}
  }

\newcommand{\mastrl}{\texttt{AST-LR}}
\newcommand{\mastec}{\texttt{AST-EC}}
\newcommand{\mastrld}{\texttt{AST-LR-d}}
\newcommand{\lrsupertree}{\texttt{LR-Consensus}}
\newcommand{\mast}{\texttt{MAST}}

\newcommand{\minrti}{\texttt{MinRTI}}

\newcommand{\mastrlval}[1]{AST_{LR}(#1)}

\newcommand{\minrtival}[1]{MINRTI(#1)}
\newcommand{\lblset}{\mathcal{X}}

\newcommand{\none}{\perp}

\renewcommand{\l}{\ell}
\newcommand{\T}{\mathcal{T}}
\newcommand{\C}{\mathcal{C}}
\newcommand{\R}{\mathcal{R}}
\renewcommand{\L}{\mathcal{L}}

\newcommand{\lca}{\textsc{lca}}

\begin{document}

\maketitle

\begin{abstract}
Understanding the evolution of a set of genes or species is a fundamental problem in evolutionary biology. The problem we study here takes as input a set of trees describing {possibly discordant} evolutionary scenarios for a given set of genes or species, and aims at finding a single tree that minimizes the leaf-removal distance to the input trees.
This problem is a specific instance of the general consensus/supertree problem, widely used  to combine or summarize discordant evolutionary trees. The problem we introduce is specifically tailored to address the case of discrepancies between the input trees due to the misplacement of individual taxa. Most supertree or consensus tree problems are computationally intractable, and we show that the problem we introduce is also NP-hard. We provide tractability results in form of a 2-approximation algorithm.
We also introduce a variant that minimizes the maximum number $d$ of leaves that are removed from any input tree, and provide a parameterized algorithm for this problem with parameter $d$. 
\end{abstract}

\setcounter{footnote}{0}


\section{Introduction}\label{sec:introduction}

In the present paper, we consider a very generic computational biology problem: given a collection of trees representing, possibly discordant, evolutionary scenarios for a set of biological entities (genes or species -- also called \emph{taxa} in the following), we want to compute a single tree that agrees as much as possible with the input trees. Several questions in computational biology can be phrased in this generic framework. For example, 
for a given set of homologous gene sequences that have been aligned, one can sample \emph{evolutionary trees} for this gene family according to a well defined posterior distribution and then ask how this collection of trees can be combined into a single gene tree, a problem known as \emph{tree amalgamation}~\cite{DBLP:journals/bioinformatics/ScornavaccaJS15}. In phylogenomics, one aims at \emph{inferring a species tree} from a collection of input trees obtained from whole-genome sequence data. A first approach considers gene families and proceeds by computing individual \emph{gene trees} from a large set of gene families, and then combining this collection of gene trees into a unique species tree for the given set of taxa; this requires handling the discordant signal observed in the gene trees due to evolutionary processes such as gene duplication and loss~\cite{10.1073/pnas.1412770112}, lateral gene transfer~\cite{10.1073/pnas.1202997109}, 
or incomplete lineage sorting~\cite{10.1093/sysbio/syw082}. Another approach concatenates the sequence data into a single large multiple sequence alignment, that is then partitioned into overlapping subsets of taxa for which partial evolutionary trees are computed, and a unique species tree is then inferred by combining the resulting collection of partial trees~\cite{10.1126/science.1253451}. 

For example, the  Maximum Agreement Subtree (MAST) problem considers a collection of input trees\footnote{All trees we consider here are uniquely leaf-labeled, rooted (\textit{i.e.} are out-trees) and binary; see next section for formal definitions.}, all having the same leaf labels and looks for a tree of maximum size (number of leaves), which agrees with 
each of the input trees. This problem is tractable for trees with bounded degree but NP-hard generally~\cite{Amir1997}. The MAST problem is a \emph{consensus problem}, because the input trees share the same leaf labels set, and the output tree is called a   \emph{consensus} tree. In the \textit{supertree framework}, the input trees might not all have identical label sets, but the output is a tree on the whole label set, called a \emph{supertree}. For example, in the Robinson-Foulds (RF) supertree problem, the goal is to find a supertree that minimizes the sum of the RF-distances to the individual input trees~\cite{10.1093/bioinformatics/btw600}. 
One way to compute consensus trees and supertrees that is closely related to our work is to modify the collection of input trees minimally in such a way that the resulting modified trees all agree. For example, in the MAST problem, modifications of the input trees consist in removing a minimum number of taxa from the whole label set,  while in the Agreement Supertree by Edge Contraction (AST-EC) problem, one is asked to contract a minimum number of edges of the input trees such that the resulting (possibly non-binary) trees all agree with at least one supertree~\cite{DBLP:journals/siamcomp/Fernandez-BacaG15}; in the case where the input trees are all triplets (rooted trees on three leaves), this supertree problem is known as the Minimum Rooted Triplets Inconsistency problem~\cite{journals/dam/byrka2010}. The SPR Supertree problem considers a similar problem where the input trees can be modified with the Subtree-Prune-and-Regraft (SPR) operator~\cite{whidden2014supertrees}.

\looseness=-1
In the present work, we introduce a new consensus problem, called \lrsupertree{}. Given a collection of input trees having the same leaf labels set, we want to remove a minimum number of leaves -- an operation called a Leaf-Removal (LR) -- from the input trees such that the resulting pruned trees all agree. Alternatively, this can be stated as finding a consensus tree that minimizes the cumulated \textit{leaf-removal distance} to the collection of input trees. This problem also applies to tree amalgamation  and to species tree inference from one-to-one orthologous gene families, where the LR operation aims at correcting the misplacement of a single taxon in an input tree.


In the next section, we formally define the problems we consider, and how they relate to other supertree problems. Next we show that the \lrsupertree{} problem is  NP-hard and that in some instances, a large number of leaves need to be removed to lead to a consensus tree. We then provide a 2-approximation algorithm\footnote{In a previous version of this work, we claimed a fixed-parameter tractable result for \lrsupertree, which turned out to be inaccurate.  See Section~\ref{sec:fpt} for details.}.
We then introduce a variant of the \lrsupertree{} problem, where we ask if a consensus tree can be obtained by removing at most $d$ leaves from each input tree, and describe a fixed-parameter tractable (FPT) algorithm with parameter $d$.

\section{Preliminary notions and problems statement}

\paragraph{Trees.} All trees in the rest of the document are assumed to be rooted and binary.
If $T$ is a tree, we denote its root by $r(T)$ and its  leaf set by $\L(T)$.  Each leaf is labeled by a distinct element 
from a \emph{label set} $\lblset$, and we denote by $\lblset(T)$ the set of labels of the leaves of $T$.
We may sometimes use $\L(T)$ and $\lblset(T)$ interchangeably.
For some $X \subseteq \lblset$, we denote by $lca_T(X)$ the \emph{least common ancestor} of $X$ in $T$. 
The subtree rooted at a node $u \in V(T)$
is denoted $T_u$ and we may write $\L_T(u)$ for $\L(T_u)$.
If $T_1$ and $T_2$ are two trees and $e$ is an edge 
of $T_1$, grafting $T_2$ on $e$ consists in subdividing $e$ and letting the resulting degree $2$ node become the parent 
of $r(T_2)$.  Grafting $T_2$ above $T_1$ consists in creating a new node $r$, then letting $r$ become the parent of $r(T_1)$ and $r(T_2)$.  Grafting $T_2$ on $T_1$ means grafting $T_2$ either on an edge of $T_1$ or above $T_1$.

\paragraph{The leaf removal operation.}
\looseness=-1
For a subset $L \subseteq \lblset$, we denote by $T - L$ the tree obtained from $T$ by removing every leaf labeled by $L$, contracting the resulting non-root vertices of degree two, and repeatedly deleting the resulting root vertex while it has degree one.
The \emph{restriction} $T|_L$ of $T$ to $L$ is the tree $T - (\lblset \setminus L)$, \textit{i.e.} the tree obtained by removing every leaf \emph{not} in $L$.
A \emph{triplet} is a rooted tree on $3$ leaves.  We denote a triplet $\triplet$ with leaf set $\{a,b,c\}$
by $ab|c$ if $c$ is the leaf that is a direct child of the root (the parent of $a$ and $b$ being its other child).
We say $\triplet = ab|c$ is a triplet of a tree $T$ if $T|_{\{a,b,c\}} = \triplet$.
We denote $tr(T) = \{ab|c : ab|c$ is a triplet of $T\}$.

We define a \emph{distance function} $d_{LR}$ between 
two trees $T_1$ and $T_2$ on the same label set $\lblset$ consisting in the minimum number of labels to remove from $\lblset$ so that the two trees are equal.  
That is, 
$$d_{LR}(T_1, T_2) = \min \{ |X| : X \subseteq \lblset \mbox{ and } T_1 - X = T_2 - X \}$$
Note that $d_{LR}$ is closely related to the Maximum Agreement Subtree (MAST) between two trees on the same label set $\lblset$, which consists in a subset $X' \subseteq \lblset$ of maximum size such that 
$T_1|_{X'} = T_2|_{X'}$:  $d_{LR}(T_1, T_2) = |\lblset| - |X'|$.  The MAST of two binary trees on the same label set can be computed in 
time $O(n \log n)$, where $n = |\lblset|$~\cite{DBLP:journals/siamcomp/ColeFHPT00}, and so $d_{LR}$ can be found within the {same} time complexity.

\paragraph{Problem statements.}
In this paper, we are interested in finding a tree $T$ on $\lblset$  
minimizing the sum of $d_{LR}$ distances to a given set of input trees.

\vspace{1mm}

\noindent \lrsupertree{} \\
\noindent {\bf Given}: a set of trees $\T = \{T_1, \ldots, T_t\}$ with  $\lblset(T_1) = \ldots = \lblset(T_t) = \lblset$. \\
\noindent {\bf Find}:  a tree $T$ on  label set $\lblset$ that minimizes $\sum_{T_i \in \T} d_{LR}(T, T_i)$.\\[-2mm]

We can reformulate the \lrsupertree{} problem as the problem 
of removing a minimum number of leaves from the input trees
so that they are \emph{compatible}.  Although the equivalence between both formulations is obvious, the later formulation will often be more convenient.  We need to introduce more definitions in order to establish this equivalence.  

A set of trees $\T = \{T_1, \ldots, T_t\}$ is called \emph{compatible} if there is a tree
$T$ such that $\lblset(T) = \bigcup_{T_i \in \T}\lblset(T_i)$ and $T|_{\lblset(T_i)} = T_i$ for every $i \in [t]$.  In this case, we say that $T$ \emph{displays} $\T$.
A list $\C = (\lblset_1, \ldots, \lblset_t)$ of subsets of $\lblset$ is a
\emph{leaf-disagreement} for $\T$ 
if $\{T_1 - \lblset_1, \ldots, T_t - \lblset_t\}$ is compatible.  The \emph{size} of $\C$ is $\sum_{i \in  [t]}|\lblset_i|$.
We denote by $\mastrlval{\T}$ the minimum size of a leaf-disagreement for $\T$, 
and may sometimes write $\mastrlval{T_1, \ldots, T_t}$ instead of $\mastrlval{\T}$.
A subset $\lblset' \subseteq \lblset$ of labels is a \emph{label-disagreement}
for $\T$ if $\{T_1 - \lblset', \ldots, T_t - \lblset'\}$ is compatible.
Note that, if $\T = \{T_1, T_2\}$, then the minimum
size of a label-disagreement for $\T$ is $d_{LR}(T_1, T_2)$.
We may now define the \mastrl{} problem
{(see Figure~\ref{fig:example} for an example)}.

\vspace{2mm}

\noindent \texttt{Agreement Subtrees by Leaf-Removals} (\mastrl{}) \\
\noindent {\bf Given}: a set of trees $\T = \{T_1, \ldots, T_t\}$ with  $\lblset(T_1) = \ldots = \lblset(T_t) = \lblset$. \\
\noindent {\bf Find}:  a leaf-disagreement $\C$ for $\T$ of minimum size.\\


\vspace*{-5mm}

\begin{figure*}[h]
\centering
\includegraphics[width=\textwidth]
{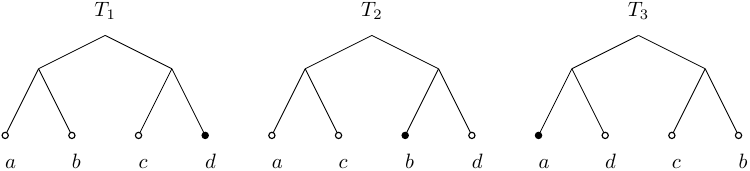}
\caption{Example instance $\T = \{T_1, T_2, T_3\}$ of \mastrl{} with label set $\lblset = \{a,b,c,d\}$. The list $(\lblset_1 = \{d\}, \lblset_2 = \{b\}, \lblset_3 = \{a\})$  is a leaf-disagreement for $\T$ of size $3$.}
\label{fig:example}
\end{figure*}

\vspace*{-7mm}

\begin{lemma}\label{lem:equiv-problems}
Let $\T = \{T_1, \ldots, T_t\}$ be a set of trees on the same  label set {$\lblset$}, with $n = |\lblset|$.
Given a supertree $T$ such that $v := \sum_{T_i \in \T}d_{LR}(T, T_i)$, one can compute in time $O(t n \log(n))$ a leaf-disagreement $\C$ of size at most $v$.
Conversely, given a leaf-disagreement $\C$ for $\T$ of size $v$, one can compute in time $O(t n \log^2 (tn))$ a supertree $T$ such that 
$\sum_{T_i \in \T}d_{LR}(T, T_i) \leq v$.
\end{lemma}


From Lemma~\ref{lem:equiv-problems}\footnote{All missing proofs are provided in Appendix. 
}
both problems share the same optimality value, 
the NP-hardness of one implies the hardness of the other and 
approximating one problem within a factor $c$ implies that the other 
problem can be approximated within a factor $c$.
We conclude this subsection with the introduction of a parameterized  variant of the \mastrl{} problem.

\vspace{2mm}

\noindent \textbf{\textsc{\mastrld{}}} \\
\noindent {\bf Input}: a set of  
trees $\T = \{T_1, \ldots, T_t\}$ with 
$\L(T_1) = \ldots = \L(T_t) = \lblset$, and an integer $d$.\\
\noindent {\bf Question}: Are there $\lblset_1, \ldots, \lblset_t \subseteq \lblset$ such that $|\lblset_i| \leq d$ for each $i \in [t]$, and $\{T_1 - \lblset_1, \ldots, T_t - \lblset_t\}$ is compatible?\\

\vspace{-2mm}
\looseness=-1
We call a tree $T^*$ a \emph{solution} to the $\mastrld{}$ instance if $d_{LR}(T_i, T^*) \leq d$ for each $i \in [t]$.





\vspace{-2mm}
\paragraph{Relation to other supertree/consensus tree problems.}
\looseness=-1
The most widely studied supertree problem based on modifying the input trees is the SPR Supertree problem, where arbitrarily large subtrees can be moved in the input trees to make them all agree (see~\cite{whidden2014supertrees} and references there).
The interest of this problem is that the SPR operation is very general, modelling lateral gene transfer and introgression. The LR operation we introduce is a limited SPR, where the displaced subtree is composed of a single leaf. An alternative to the SPR operation to move subtrees within a tree is the Edge Contraction (EC) operation, that contracts an edge of an input tree, thus increasing the degree of the parent node.
This operation allows correcting the local misplacement of a full subtree.
\mastec{} is NP-complete but can be solved in $O((2t)^ptn^2)$~time where $p$ is the number of required EC operations~\cite{DBLP:journals/siamcomp/Fernandez-BacaG15}.

Compared to the two problems described above, an LR models a very specific type of error in evolutionary trees, that is the misplacement of a single taxon (a single leaf) in one of the input trees. This error occurs frequently in reconstructing evolutionary trees, and can be  caused for example by some evolutionary process specific to the corresponding input tree (recent incomplete lineage sorting, or recent lateral transfer for example). Conversely, it is not well adapted to model errors, due for example to ancient evolutionary events that impacts large subtrees. However, an attractive feature of the LR operation is that computing the LR distance is  equivalent to computing the \mast{} cost and is thus tractable, unlike the SPR distance which is hard to compute. This suggests that the \lrsupertree{} problem might be easier to solve than the SPR Supertree problem, and we provide indeed several tractability results. Compared to the \mastec{} problem, the \mastrl{} problem is naturally more adapted to correct single taxa misplacements as the EC operation is very local and  the number of EC required to correct a taxon misplacement is linear  in the length of the path to its correct location, while the LR cost of correcting this is unitary. Last,  \lrsupertree{} is more flexible than the \mast{} problem as it relies on modifications of the input trees, while with the way \mast{} corrects  a misplaced leaf requires to remove this leaf from all input trees. This shows that the problems \mastrl{} and \mastrld{} complement well the existing corpus of gene trees correction models.


\section{Hardness and approximability of \mastrl{}}\label{sec:hardness}

In this section, we show that the $\mastrl{}$ problem is NP-hard, from which the \lrsupertree{} hardness follows.  We then describe a simple factor $2$ approximation algorithm.  The algorithm turns out to be useful for analyzing the 
worst case scenario for $\mastrl{}$ in terms of the required number of leaves to remove, as we show that there are $\mastrl{}$ instances that require removing about $n - \sqrt{n}$ leaves in each input tree.

\subsection*{NP-hardness of \mastrl{}}

We assume here that we are considering the decision version of $\mastrl{}$, \textit{i.e.} deciding whether there is a leaf-disagreement of size at most $\l$ for a given $\l$. We use a reduction from the \minrti{} problem: given a set $\R$ of rooted triplets, find a subset $\R' \subset \R$ of minimum cardinality such that $\R \setminus \R'$ is compatible.  The \minrti{} problem is NP-Hard~\cite{journals/dam/byrka2010} (even $W[2]$-hard and  hard to 
approximate within a $O(\log n)$ factor). Denote by $\minrtival{\R}$ the minimum number of triplets to remove from $\R$ to attain compatibility. 
We describe the reduction here.

Let $\R = \{R_1, \ldots, R_t\}$ be an instance of $\minrti$,  with the label set $L := \bigcup_{i = 1}^t \lblset(R_i)$.  For a given integer $m$, we construct an \mastrl{} instance 
$\T = \{T_1, \ldots, T_t\}$ which is such that $MINRTI(\R) \leq m$ if and only if 
$\mastrlval{\T} \leq t(|L| - 3) + m$.

We first construct a tree $Z$ with additional labels which will serve as our main gadget.
Let $\{L_i\}_{1 \leq i \leq t}$ be a collection of $t$  new  label sets, each of size $(|L|t)^{10}$, all disjoint from each other and all disjoint from $L$.  Each tree in our $\mastrl{}$ instance will be on label set $\lblset = L \cup L_1 \cup \ldots \cup L_t$. 
For each $i \in [t]$, let $X_i$ be any tree with  label set $L_i$.
Obtain $Z$ by taking any tree on $t$ leaves $l_1, \ldots, l_t$, then replacing each leaf $l_i$
by the $X_i$ tree (\textit{i.e.} $l_i$ is replaced by $r(X_i)$).  Denote by $r_Z(X_i)$ the root of the $X_i$ subtree in $Z$.

Then for each $i \in [t]$, we construct $T_i$ from $R_i$ as follows.
Let $L' = L \setminus \lblset(R_i)$ be the set of labels not appearing in $R_i$, noting that $|L'| = |L| - 3$.
Let $T_{L'}$ be any tree with  label set   $L'$, and obtain the tree $Z_i$ by grafting 
$T_{L'}$ on the edge between $r_Z(X_i)$ and its parent.
Finally, $T_i$ is obtained by grafting $R_i$ above $Z_i$.
{See Figure~\ref{fig:hardness} for an example.}
Note that each tree $T_i$ has  label set $\lblset$ as desired.
Also, it is not difficult to see that this reduction can be carried out in polynomial time.  
This construction can now be used to show the following.

\begin{figure*}[t]
\centering
\includegraphics[width=0.35\textwidth]{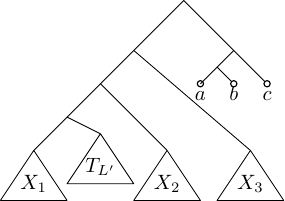}
\caption{Construction of the tree $T_1$ for an instance $\R = \{R_1, R_2, R_3\}$ of $\minrti$ in which $R_1 = ab|c$.}
\label{fig:hardness}
\end{figure*}

\begin{theorem}\label{thm:np-hard}
The \mastrl{} and \lrsupertree{} problems are NP-hard.
\end{theorem}

{The idea of the proof is to show that in the constructed \mastrl{} instance, we are "forced" to solve the corresponding \minrti{} instance. In more detail, we show that $MINRTI(\R) \leq m$ if and only if 
$\mastrlval{\T} \leq t(|L| - 3) + m$.
In one direction, given a set $\R'$ of size $m$ such that $\R\setminus\R'$ is compatible, one can show that the following leaf removals from $\T$ make it compatible: 
remove, from each $T_i$, the leaves $L' = L \setminus \lblset(R_i)$ that were inserted into the $Z$ subtree, then for each $R_i \in \R'$, remove a single leaf in $\lblset(R_i)$ from $T_i$. 
This sums up to $t(|L| - 3) + m$ leaf removals.
Conversely, it can be shown that there always exists an optimal solution for $\T$ that removes, for each $T_i$, all the leaves $L' = L \setminus \lblset(R_i)$ inserted in the $Z$ subtree, plus an additional single leaf $l$ from $m$ trees $T_{i_1}, \dots, T_{i_m}$ such that $l \in L$.
The corresponding triplets $R_{i_1},\dots, R_{i_m}$ can be removed from $\R$ so that it becomes compatible.}


\subsection*{Approximating \mastrl{} and bounding worst-case scenarios}

Given the above result, it is natural to turn to approximation algorithms in order to solve \mastrl{} or \lrsupertree{} instances.
It turns out that there is a simple factor $2$ approximation for \lrsupertree{} which is achieved by interpreting the problem 
as finding a median in a metric space.
Indeed, it is not hard to see that $d_{LR}$ is a metric (over the space of trees on the same  label set $\lblset$).
A direct consequence, using an argument akin to the one in~\cite[p.351]{books/gusfield1997}, is the following.

\begin{theorem}\label{lem:2-approx}
The following is a factor $2$ approximation algorithm for \lrsupertree: return the tree $T \in \T$ that minimizes 
$\sum_{T_i \in \T}d_{LR}(T, T_i)$.
\end{theorem}

Theorem~\ref{lem:2-approx} can be used to lower-bound the 
`worst' possible instance of \mastrl{}.
We show that in some cases, we can only keep 
about $\sqrt{|\lblset|}$ leaves per tree.  That is, there are instances for which  
$\mastrlval{\T} = \Omega(t(n - \sqrt{n}))$, where $t$ is the number of trees and $n = |\lblset|$.  The argument is based on a probabilistic argument, for which we will make use of the following result~\cite[Theorem 4.3.iv]{bryant2003size}.

\begin{theorem}[\cite{bryant2003size}]\label{thm:expectedmast}
For any constant $c > e/\sqrt{2}$, there is some $n_0$ such that for all $n \geq n_0$, the following holds:
if $T_1$ and $T_2$ are two binary trees on $n$ leaves chosen randomly, uniformly and independently, then $\mathbb{E}[d_{LR}(T_1, T_2)] \geq n - c\sqrt{n}$.  
\end{theorem}

\begin{corollary}\label{thm:worst-case}
There are instances of \mastrl{} in which $\Omega(t(n - \sqrt{n}))$
leaves need to be deleted.
\end{corollary}

\looseness=-1
The above is shown by demonstrating that, by picking a set $\T$ of $t$ random trees, the expected optimal sum of distances $\min_T \sum_{T_i \in \T} d_{LR}(T, T_i)$ is $\Omega(t(n - \sqrt{n})$.  This is not direct though, since the tree $T^*$ that minimizes this sum is not itself random, and so we cannot apply Theorem~\ref{thm:expectedmast} directly on $T^*$.
We can however, show that the tree $T' \in \T$ obtained using the 2-approximation, which is random, has expected sum of distances $\Omega(t(n - \sqrt{n}))$.  Since $T^*$ requires, at best, half the leaf deletions of $T'$, the result follows. Note that finding a non-trivial upper bound on $\mastrlval{\T}$ is open.

\section{Fixed-parameter tractability of \mastrl{} and \mastrld{}.}\label{sec:fpt}

An alternative way to deal with computational hardness is parameterized complexity.
The most natural parameter to study is $q := \mastrlval{\T}$, the question being whether there exists an algorithm for \mastrl{} that runs in time $O(f(q) poly(n))$ for some function $f$ depending only on $q$.  In a previous version of this work, we proposed an FPT algorithm that turned out to be inaccurate.  This was observed by Chen et al.~\cite{chen2019computing}, who proposed a fix through an alternate proof of FPT membership.  We invite the interested reader to consult the second arXiv version [V2] of this paper in Section 4 for the details of the inaccurate algorithm.

We consider an alternative parameter $d$, and show that finding a tree $T^*$, if it exists, such that $d_{LR}(T_i, T^*) \leq d$ for every input tree $T_i$, can be done in $O(c^d d^{3d}(n^3 + tn \log n))$~time for some constant $c$.

\subsection{Parameterization by maximum distance $d$}

We now describe an algorithm for the \mastrld{} problem, running in time $O(c^d d^{3d}(n^3 + tn \log n))$ that, if it exists, finds a solution (where here $c$ is a constant not depending on $d$ nor $n$).

We employ the following branch-and-bound strategy, keeping a candidate solution at each step. Initially, the candidate solution is the input tree~$T_1$ and, if $T_1$ is indeed a solution, we return it.
Otherwise (in particular if $d_{LR}(T_1,T_i) > d$ for some input tree $T_i$), we branch on a set of ``leaf-prune-and-regraft'' operations on $T_1$. In such an operation, we prune one leaf from $T_1$ and regraft it somewhere else. If we have not produced a solution after $d$ such operations, then we halt this branch of the algorithm (as any solution must be reachable from $T_1$ by at most $d$ operations).
The resulting search tree has depth at most $d$. In order to bound the running time of the algorithm, we need to bound the number of  ``leaf-prune-and-regraft'' operations to try at each branching step. There are two steps to this: first, we bound the set of candidate leaves to prune, second, given a leaf, we bound the number of places where to regraft it.
To bound the candidate set of leaves to prune, let us call a leaf~$x$ \emph{interesting} if there is a solution $T^*$, and minimal sets $X_1,X_i \subseteq \lblset$ of size at most $d$, such that
\begin{inparaenum}[(a)]
\item $T_1 - X_1 = T^* - X_1$,
\item $T_i - X_i = T^* - X_i$, and
\item $x \in X_1 \setminus X_i$,
\end{inparaenum}
where $T_i$ is an arbitrary input tree for which $d_{LR}(T_1,T_i) > d$.
It can be shown that an interesting leaf~$x$ must exist if there is a solution.
Moreover, though we cannot identify $x$ before we know $T^*$, we can nevertheless construct a set~$S$ of size~$O(d^2)$ containing all interesting leaves.
Thus, in our branching step, it suffices to consider leaves in $S$.

{Assuming we have chosen the correct $x$, we then bound the number of places to try regrafting $x$.
Because of the way we chose $x$, we may assume there is a solution $T^*$ and $X_i \subseteq \lblset$ such that $|X_i| \leq d$, $T_i - X_i = T^* - X_i$ and $x \notin X_i$. Thus we may treat $T_i$ as a ``guide'' on where to regraft $x$. Due to the differences between $T_1$, $T_i$ and $T^*$, this guide does not give us an exact location in $T_1$ to regraft $x$. Nevertheless, we can show that the number of candidate locations to regraft $x$ can be bounded by $O(d)$.
Thus, in total we have $O(d^3)$ branches at each step in our search tree of depth $d$, and therefore have to consider $O((O(3^d))^{d}) = O(c^dd^{3d})$ subproblems.}

{
\begin{theorem}\label{thm:fpt-in-d}
\mastrld{} can be solved in time $O(c^d d^{3d}(n^3 + tn \log n))$, where  $c$ is a constant not depending on $d$ or $n$.
\end{theorem}
}

\vspace{-3mm}

\section{Conclusion}

\vspace{-1mm}

To conclude, we introduced a new supertree/consensus problem, based on a simple combinatorial operator acting on trees, the Leaf-Removal. We showed that, although this supertree problem is NP-hard, it admits interesting tractability results, that compare well with existing algorithms. Future research should explore if various simple combinatorial operators, that individually define relatively tractable supertree problems (for example LR and EC) can be combined into a unified supertree problem while maintaining approximability and fixed-parameter tractability.

\vspace{1mm}

{\small {\bf Acknowledgement:}
MJ was partially supported by Labex NUMEV (ANR-10-LABX-20) and Vidi grant 639.072.602 from The Netherlands Organization for Scientific Research (NWO).
CC was supported by NSERC Discovery Grant 249834.
CS was partially supported by the French Agence Nationale de la Recherche Investissements d’Avenir/Bioinformatique (ANR-10-BINF-01-01, ANR-10-BINF-01-02, Ancestrome).
ML was supported by NSERC PDF Grant.
MW was supported by the Institut de Biologie Computationnelle.
}

\bibliographystyle{plain}
\bibliography{mast}

\vfill\pagebreak
\appendix 
\section{Omitted proofs}

{Here we give proofs for  several results whose  proofs were omitted in the main paper. Note that the proof of Theorem~\ref{thm:fpt-in-d} is deferred to its own section.}




\medskip
 \noindent
 \textbf{Lemma~\ref{lem:equiv-problems}}
 \emph{
 (restated).
Let $\T = \{T_1, \ldots, T_t\}$ be a set of trees on the same  label set $\lblset$.
Then, given a supertree $T$ such that $v := \sum_{T_i \in \T}d_{LR}(T, T_i)$, one can compute in time $O(t n \log n)$ a leaf-disagreement $\C$ of size at most $v$, where $n = |\lblset|$.
Conversely, given a leaf-disagreement $\C$ for $\T$ of size $v$, one can compute in time $O(t n \log^2 (tn))$ a supertree $T$ such that 
$\sum_{T_i \in \T}d_{LR}(T, T_i) \leq v$.
}

\begin{proof}
In the first direction, for each $T_i \in \T$, there is 
a set $X_i \subseteq \lblset$ of size $d_{LR}(T, T_i)$ such that $T_i - X_i = T - X_i$.  Moreover, $X_i$ can be found in time $O(n \log n)$.  Thus $(X_1, \ldots, X_t)$
is a leaf-disagreement of the desired size and can be found in time $O(t n \log n)$.
Conversely, let $\C = (X_1, \ldots, X_t)$ be a leaf-disagreement of size $v$.  
As $\T' = \{T_1 - X_1, \ldots, T_t - X_t\}$ is compatible, 
there is a tree $T$ that displays $\T'$, and it is easy to see that the sum of distances between $T$ and $\T'$
is at most the size of $\C$.  As for the complexity, 
it is shown in~\cite{deng_et_al:LIPIcs:2016:6088} how to compute in time $O(tn \log^2 (tn))$, given a set of trees $\T'$, a tree 
$T$ displaying $\T'$ if one exists.
\qed
\end{proof}

We next consider the case where $\T$ consists only of two trees.

\begin{lemma}\label{lem:two-trees}
Let $T_1, T_2$ be two trees on the same  label set $\lblset$.
Then $\mastrlval{T_1, T_2} = d_{LR}(T_1, T_2)$. 
Moreover, every optimal leaf-disagreement $\C = (\lblset_1', \lblset_2')$ for $T_1$ and $T_2$ can be obtained 
in the following manner: 
for every label-disagreement $\lblset'$ of size $d_{LR}(T_1, T_2)$, partition $\lblset'$ into $\lblset'_1, \lblset'_2$.
\end{lemma}
\begin{proof}
Let $\lblset' \subset \lblset$ such that  $|\lblset'| = d_{LR}(T_1, T_2)$
and $T_1 - \lblset' = T_2 - \lblset'$.
Then clearly, for any bipartition $(\lblset'_1, \lblset'_2)$ of $\lblset'$, 
$T_1' := T_1 - \lblset'_1$ and $T_2' := T_2 - \lblset'_2$ are compatible, since the leaves that $T_1'$ and $T_2'$ have in common yield the same subtree, and leaves that appear in only one tree cannot create incompatibility.  In particular, $\mastrlval{T_1, T_2} \leq d_{LR}(T_1, T_2)$.

Conversely, let $\C = (\lblset'_1, \lblset'_2)$ be a minimum 
leaf-disagreement.  
We have $\lblset'_1 \cap \lblset'_2 = \emptyset$, for if there is some $\l \in  \lblset'_1 \cap \lblset'_2$, then $\l$ could be reinserted into one of the two trees without creating incompatibility.
Thus $\C$ is a bipartition of $\lblset' = \lblset'_1 \cup \lblset'_2$.
Moreover, we must have $T_1 - \lblset' = T_2 - \lblset'$, implying 
$|\lblset'| \geq d_{LR}(T_1, T_2)$.  Combined with the above inequality, 
$|\lblset'| = d_{LR}(T_1, T_2)$, and the Lemma follows.
\end{proof}

It follows from Lemma~\ref{lem:two-trees} that any optimal label-disagreement $\lblset'$ can be turned into an optimal leaf-disa\-greement, which is convenient as $\lblset'$ can be found in polynomial time.
 We will make heavy use of this property later on.
 
 Note that the same type of equivalence does not hold when $3$ or more trees are given, \textit{i.e.} computing a MAST of three trees does not 
 necessarily yield a leaf-disagreement of minimum size.
 Consider for example the instance $\T = \{T_1,T_2,T_3\}$ in Figure~\ref{fig:example}. An optimal leaf-disagreement for $\T$ has size $2$ and consists of any pair of distinct leaves. On the other hand, an optimal leaf-disagreement for $\T$ has size $3$, and moreover each leaf corresponds to a different label.




\medskip
 \noindent
 \textbf{Theorem~\ref{thm:np-hard}}
 \emph{
 (restated).
The \mastrl{} and \lrsupertree{} problems are NP-hard.
}

\begin{proof}
We begin by restating the reduction from $\minrti$ to \mastrl.

Let $\R = \{R_1, \ldots, R_t\}$ be an instance of $\minrti$,  with the label set $L := \bigcup_{i = 1}^t \lblset(R_i)$.  For a given integer $m$, we construct an \mastrl{} instance 
$\T = \{T_1, \ldots, T_t\}$ which is such that $MINRTI(\R) \leq m$ if and only if 
$\mastrlval{\T} \leq t(|L| - 3) + m$.

We first construct a tree $Z$ with additional labels which will serve as our main gadget.
Let $\{L_i\}_{1 \leq i \leq t}$ be a collection of $t$  new  label sets, each of size $(|L|t)^{10}$, all disjoint from each other and all disjoint from $L$.  Each tree in our $\mastrl{}$ instance will be on label set $\lblset = L \cup L_1 \cup \ldots \cup L_t$. 
For each $i \in [t]$, let $X_i$ be any tree with  label set $L_i$.
Obtain $Z$ by taking any tree on $t$ leaves $l_1, \ldots, l_t$, then replacing each leaf $l_i$
by the $X_i$ tree (\textit{i.e.} $l_i$ is replaced by $r(X_i)$).  Denote by $r_Z(X_i)$ the root of the $X_i$ subtree in $Z$.

Then for each $i \in [t]$, we construct $T_i$ from $R_i$ as follows.
Let $L' = L \setminus \lblset(R_i)$ be the set of labels not appearing in $R_i$, noting that $|L'| = |L| - 3$.
Let $T_{L'}$ be any tree with  label set   $L'$, and obtain the tree $Z_i$ by grafting 
$T_{L'}$ on the edge between $r_Z(X_i)$ and its parent.
Finally, $T_i$ is obtained by grafting $R_i$ above $Z_i$.
{See Figure~\ref{fig:hardness} for an example.}
Note that each tree $T_i$ has  label set $\lblset$ as desired.
Also, it is not difficult to see that this reduction can be carried out in polynomial time.

We now show that $MINRTI(\R) \leq m$ if and only if 
$\mastrlval{\T} \leq t(|L| - 3) + m$.

($\Rightarrow$) 
Let $\R' \subset \R$ such that $|\R'| \leq m$ and $\R^* := \R \setminus \R'$ is compatible, and 
let $T(\R^*)$ be a tree displaying $\R^*$. Note that $|\R^*| \geq t - m$.
We obtain a \mastrl{} solution by first deleting, in each $T_i \in \T$, all the leaves labeled by $L \setminus \lblset(R_i)$ (thus $T_i$ becomes the tree obtained by grafting $R_i$ above $Z$). 
Then for each deleted triplet $R_i \in \R'$, we remove any single leaf of $T_i$ labeled by some element in $\lblset(R_i)$.  In this manner, no more than $t(|L| - 3) + m$ leaves get deleted.
Moreover, grafting $T(\R^*)$ above $Z$ yields a tree displaying the modified set of trees, showing that they are compatible.

($\Leftarrow$) 
We first argue that if $\T$ admits a leaf-disagreement $\C = (\lblset_1, \ldots, \lblset_t)$ of size at most $t(|L| - 3) + m$, 
then there is a better or equal solution that removes, in each $T_i$, all the leaves labeled 
by $L \setminus \lblset(R_i)$ (\textit{i.e.} those grafted in the $Z_i$ tree). For each $i \in [t]$, let $T_i' = T_i - \lblset_i$, and denote $\T' = \{T_1', \ldots, T_t'\}$.
Suppose that there is some $i \in [t]$ and some $\l \in L \setminus \lblset(R_i)$ such that $\l \in \lblset(T_i')$.  

We claim that $\l \notin \lblset(T_j')$ for every $i \neq j \in [t]$.  Suppose otherwise that $\l \in \lblset(T_j')$ for some $j \neq i$.  
Consider first the case where $\l \notin \lblset(R_j)$.
Note that by the construction of $Z_i$ and $Z_j$, for every $x_i \in \lblset(X_i) \cap \lblset(T_i') \cap \lblset(T_j')$ and every $x_j \in \lblset(X_j) \cap \lblset(T_i')  \cap \lblset(T_j')$, 
$T_i'$ contains the $\l x_i | x_j$ triplet whereas $T_j'$ contains the $\l x_j | x_i$ triplet.  Since these triplets are conflicting, no supertree can contain both and so no such $x_i, x_j$ pair can exist, as we are assuming that a supertree for $T'_i$ and $T'_j$ exists.  This implies that one of $\lblset(X_i) \cap \lblset(T_i') \cap \lblset(T_j')$ or $\lblset(X_j) \cap \lblset(T_i')  \cap \lblset(T_j')$ must be empty.  Suppose without loss of generality that the former is empty.  Then each $x_i \in X_i$ must have been deleted in at least one of $T_i$ or $T_j$.  As $|\lblset(X_i)| = (|L|t)^{10} > t(|L| - 3) + m$, this contradicts the size of the 
solution $\C$.  In the second case, we have $\l \in \lblset(R_j)$.  But this time, if there are $x_i \in \lblset(X_i) \cap \lblset(T_i') \cap \lblset(T_j')$ and $x_j \in \lblset(X_j) \cap \lblset(T_i') \cap \lblset(T_j')$, then $T_j'$ contains the $x_ix_j|\l$ triplet, again conflicting with the $\l x_i | x_j$ triplet found in $T_i$.  As before, we run into a contradiction since too many $X_i$ or $X_j$ leaves need to be deleted. This proves our claim. 

We thus assume that $\l$ only appears in $T_i'$.
Let $R_j \in \R$ such that $\l \in \lblset(R_j)$, noting that
$\l$ does not appear in $T'_j$. 
Consider the solution $\T''$ obtained from $\T'$ by removing $\l$ from $T_i'$, 
and placing it back in $T_j'$ where it originally was in $T_j$.  Formally this is achieved by replacing, in the leaf-disagreement $\C$, 
$\lblset_i$ by $\lblset_i \cup \{\l\}$
and $\lblset_j$ by 
$\lblset_j \setminus \{\l\}$. 
Since $\l$ still appears only in one tree, no conflict is created and we obtain another solution 
of equal size.
By repeating this process for every such leaf $\l$, we obtain a solution in which 
every leaf labeled by $L \setminus \lblset(R_i)$ is removed from $T_i'$.
We now assume that the solution $\T'$ has this form.

Consider the subset $\R' = \{R_i \in \R : |\lblset(T_i') \cap \lblset(R_i)| < 3\}$, 
that is those triplets $R_i$ for which the corresponding tree $T_i$ had a leaf removed
outside of the $Z_i$ tree.
By the form of the $\T'$ solution, at least $t(|L| - 3)$ removals are done in the $Z_i$
trees, and as only $m$ removals remain, $\R'$ has size at most $m$.  We show that $\R \setminus \R'$ is a compatible set of triplets.  Since $\T'$ is
compatible, there is a tree $T$ that displays each $T_i' \in \T'$, and since 
each triplet of $\R \setminus \R'$ belongs to some $T_i'$, $T$ also displays
$\R \setminus \R'$.  This concludes the proof.
\qed
\end{proof}

\medskip
 \noindent
 \textbf{Theorem~\ref{lem:2-approx}}
 \emph{
 (restated).
The following is a factor $2$ approximation algorithm for \lrsupertree: return the tree $T \in \T$ that minimizes 
$\sum_{T_i \in \T}d_{LR}(T, T_i)$.
}
\begin{proof}
Let $T^*$ be an optimal solution for \lrsupertree{}, \textit{i.e.} $T^*$ is a tree minimizing $\sum_{T_i \in \T}d_{LR}(T_i, T^*)$, and let $T$ be chosen as described in the theorem statement.
Moreover let $T'$ be the tree of $\T$ minimizing $d_{LR}(T', T^*)$.
By the triangle inequality, 
$$
\sum_{T_i \in \T}d_{LR}(T', T_i) \leq \sum_{T_i \in \T}\left(  d_{LR}(T', T^*) + d_{LR}(T^*, T_i) \right) \leq 2 \sum_{T_i \in \T}d_{LR}(T^*, T_i)
$$
where the last inequality is due to the fact that $d_{LR}(T', T^*) \leq d_{LR}(T^*, T_i)$ for all $i$, by our choice of $T'$.
Our choice of $T$ implies
$\sum_{T_i \in \T}d_{LR}(T, T_i) \leq \sum_{T_i \in \T}d_{LR}(T', T_i) \leq 
2 \sum_{T_i \in \T}d_{LR}(T_i, T^*)$.
\qed
\end{proof}

\medskip
\noindent
 \textbf{Corollary~\ref{thm:worst-case}}
 \emph{
 (restated).
There are instances of \mastrl{} in which $\Omega(t(n - \sqrt{n}))$
leaves need to be deleted.
}
\begin{proof}
Let $\T = \{T_1, \ldots, T_t\}$ be a random set of $t$ trees chosen uniformly and independently.
For large enough $n$, the expected sum of distances between each pair of trees is 
$$\mathbb{E}\left[\sum_{1 \leq i < j \leq t}d_{LR}(T_i, T_j)\right] = 
\sum_{1 \leq i < j \leq t} \mathbb{E}[d_{LR}(T_i, T_j)] \geq {t \choose 2} (n - c\sqrt{n})$$
for some constant $c$, by Theorem~\ref{thm:expectedmast}.
Let $S := \min_{T} \sum_{i = 1}^t d_{LR}(T, T_i)$ be the random variable corresponding to the minimum sum of distances.
By Theorem~\ref{lem:2-approx}, there is a tree $T' \in \T$ such that $\sum_{i = 1}^t d_{LR}(T', T_i) \leq 2S$.
We have 

\begin{align*}
\sum_{1 \leq i < j \leq t}d_{LR}(T_i, T_j) &\leq \sum_{1 \leq i < j \leq t} d_{LR}(T_i, T') + d_{LR}(T', T_j) \\
&= (t - 1) \sum_{i = 1}^t d_{LR}(T_i, T') \\
&\leq (t - 1)2S
\end{align*}

Since, in general for two random variables $X$ and $Y$, 
always having $X \leq Y$ implies $\mathbb{E}[X] \leq \mathbb{E}[Y]$, we get 
$${t \choose 2}(n - c\sqrt{n}) \leq \mathbb{E}\left[\sum_{1 \leq i < j \leq t}d_{LR}(T_i, T_j)\right] \leq 
\mathbb{E}[(t - 1)2S] = 2(t - 1)\mathbb{E}[S]$$

yielding $\mathbb{E}[S] \geq t/4 (n - c\sqrt{n}) = \Omega(t(n - \sqrt{n}))$, and so there 
must exist an instance $\T$ satisfying the statement.
\end{proof}

\section{Leaf Prune-and-Regraft Moves}

Here we introduce the notion of leaf prune-and-regraft (LPR) moves, which will be used in the proof of Theorem~\ref{thm:fpt-in-d}, and which may be of independent interest.
In an LPR move, we prune a leaf from a tree and then regraft it another location (formal definitions below).
LPR moves provide an alternate way of characterizing the distance function $d_{LR}$ - indeed, we will show that $d_{LR}(T_1, T_2) \leq k$ if and only if there is a sequence of at most $k$ LPR moves transforming $T_1$ into $T_2$.


\begin{definition}
Let $T$ be a tree on  label set $\lblset$.
A LPR move on $T$ is a pair $(\l, e)$ where $\l \in \lblset$ and $e \in \{E(T - \{\l\}), \none \}$.
Applying $(\l, e)$ consists in grafting $\l$
on the $e$ edge of $T - \{\l\}$ if $e \neq \none$, and above the root of $T - \{\l\}$ if $e = \none$.

An \emph{LPR sequence} $L = ((\l_1, e_1), \ldots, (\l_k, e_k))$ is an ordered tuple of LPR moves, where for each $i \in [k]$, $(\l_i, e_i)$ is an LPR move on the tree obtained after applying the first $i - 1$ LPR moves of $L$.
We may write $L = (\l_1, \ldots, \l_k)$ if the location at which the grafting takes place 
does not need to be specified.
We say that $L$ turns $T_1$ into $T_2$ if, by applying each LPR move of $L$ in order on $T_1$,
we obtain $T_2$.
\end{definition}

See Figure~\ref{fig:LPRsequence} for an example of an LPR sequence.

\begin{figure}%
\centering
\subfloat{\includegraphics[width = 0.4\textwidth]{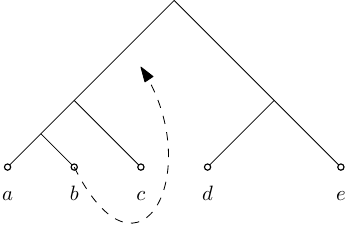}}\qquad
\subfloat{\includegraphics[width = 0.4\textwidth]{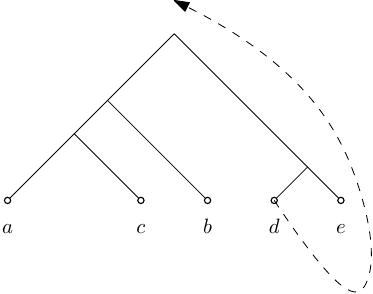}}\qquad
\subfloat{\includegraphics[width = 0.4\textwidth]{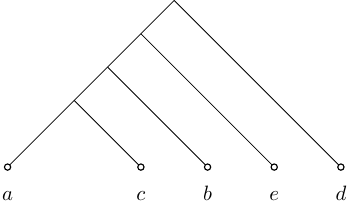}}\qquad
\caption{Sequence of trees showing the LPR sequence $L = ((b,f), (d, \none))$, where $f$ is the edge between the root and the least common ancestor of $a$ and $c$ in the first tree.}
\label{fig:LPRsequence} 
\end{figure}

In the following statements, we assume that $T_1$ and $T_2$ are two trees on  label set $\lblset$.
We exhibit an equivalence between leaf removals and LPR sequences, 
then show that the order of LPR moves in a sequence do not matter in terms of turning one tree into another - in particular any
leaf can be displaced first.

\begin{lemma}\label{lem:lpr-seq-equiv}
There is a subset $X \subseteq \lblset$ such that $T_1 - X = T_2 - X$
if and only if 
there is an LPR sequence $(x_1, x_2, \ldots, x_k)$ turning $T_1$ into $T_2$ such that $X = \{x_1, \ldots, x_k\}$.
\end{lemma}

\begin{proof}
If $T_1 = T_2$ then the proof is trivial, so we will assume this is not the case. 
We prove the lemma by induction on $|X|$.

For the base case, 
 suppose that $X = \{x\}$.  
If $T_1 - X = T_2 - X$, 
then let $T_m = T_1 - X = T_2 - X$.
We find an LPR move $(x, e)$ with $e \in E(T_m) \cup \{\none \}$ turning $T_1$ into $T_2$.
Observe that $T_2$ can be obtained by grafting $x$ on $T_m$, either on an edge $uv$, in which case we set $e = uv$, or above the root, in which case we set $e = \none$.
Since $T_m = T_1 - \{x\}$, it follows that $(x, e)$ is an LPR move turning $T_1$ into $T_2$.
In the other direction, assume there is an LPR move $(x,e)$ turning $T_1$ into $T_2$. Observe that for any tree $T'$ derived from $T_1$ by an LPR move using $x$, $T' - \{x\} = T_1 - \{x\}$.
 In particular, $T_2 - \{x\} = T_1 - \{x\}$ and we are done.
 
 For the induction step, assume that $|X|> 1$ and that the claim holds for any $X'$ such that $|X'| < |X|$.
 If $T_1 - X = T_2 - X$, then define $T_m = T_1 - X$, and let $x$ be an arbitrary element of $X$. 
 We will first construct a tree $T_1'$ such that $T_1 - \{x\} = T_1' - \{x\}$ and $T_1' - (X\setminus \{x\}) = T_2 -  (X\setminus\{x\})$.

 Observe that $T_2 - (X\setminus\{x\})$ can be obtained by grafting $x$ in $T_m$.  Let $e = uv$ if this grafting takes place on an edge of $T_m$ with $v$ being the child of $u$, or $e = \none$ if $x$ is grafted above $T_m$, and in this case let $v = r(T_m)$.  
 Let $v' = v_{T_1 - \{x\}}$ be the node in $T_1 - \{x\}$ corresponding to $v$. 

 Let $T_1'$ be derived from $T_1 - \{x\}$ by grafting $x$ onto the edge between $v'$ and its parent if $v'$ is non-root, and grafting above $v'$ otherwise.
 It is clear that  $T_1 - \{x\} = T_1' - \{x\}$.
 Furthermore,  by our choice of $v'$ we have that 
 $T_1' - (X\setminus \{x\}) = T_2 -  (X\setminus\{x\})$.
 
 Now that we have $T_1 - \{x\} = T_1' - \{x\}$ and $T_1' - (X\setminus \{x\}) = T_2 -  (X\setminus\{x\})$, by the inductive hypothesis there is 
 an LPR sequence turning $T_1$ into $T_1'$ consisting of a single move $(x,e)$,
and an LPR sequence  $(x_1, x_2, \ldots, x_{k-1})$ turning $T_1'$ into $T_2$
such that $\{x_1, \ldots, x_{k'}\} = (X\setminus \{x\})$.
Then by concatenating these two sequences, we have an LPR sequence $(x_1, x_2, \ldots, x_k)$ turning $T_1$ into $T_2$
such that $X = \{x_1, \ldots, x_k\}$.

For the converse, suppose that there is an LPR sequence $(x_1, x_2, \ldots, x_k)$ turning $T_1$ into $T_2$ such that $X = \{x_1, \ldots, x_k\}$. Let $T_1'$ be the tree derived from $T_1$ by applying the first move in this sequence. That is, there is an LPR move $(x_1,e)$ turning $T_1$ into $T_1'$, and there is an LPR sequence $(x_2, \ldots, x_k)$ turning $T_1'$ into $T_2$.
Then by the inductive hypothesis $T_1 - \{x_1\} = T_1' - \{x_1\}$ and $T_1' - \{x_2, \ldots, x_k\} = T_2 - \{x_2, \ldots, x_k\}$. Thus, $T_1- X = T_1' - X = T_2 - X$, as required.
\qed
\end{proof}

\begin{lemma}\label{lem:lpr-order}
If there is an LPR sequence $L = (x_1, \ldots, x_k)$ turning $T_1$ into $T_2$, 
then for any $i \in [k]$, there is an LPR sequence $L' = (x'_1, \ldots, x'_k)$
turning $T_1$ into $T_2$ such that $x'_1 = x_i$ and $\{x_1, \ldots, x_k\} = \{x'_1, \ldots, x'_k\}$.
\end{lemma}

\begin{proof}
Consider again the proof  that if $T_1 - X = T_2 - X$ then there is an LPR sequence $(x_1, \dots x_k)$ turning $T_1$ into $T_2$ such that $X = \{x_1, \dots, x_k\}$ (given in the proof of Lemma~\ref{lem:lpr-seq-equiv}).
When $|X| > 1$, we construct this sequence by concatenating the LPR move $(x,e)$ with an LPR sequence of length $|X|-1$, where $x$ is an arbitrary element of $X$.
As we could have chosen any element of $X$ to be $x$, we have the following: 
If $T_1 - X = T_2 - X$ then for each $x\in X$, there is an LPR sequence $(x_1, \dots, x_k)$ turning $T_1$ into $T_2$ such that $X = \{x_1, \dots, x_k\}$ and $x_1 = x$.

Thus our proof is as follows: Given an LPR sequence $L = (x_1, \ldots, x_k)$ 
turning $T_1$ into $T_2$ and some $i \in [k]$, Lemma~\ref{lem:lpr-seq-equiv} implies that $T_1 - \{x_1, \ldots, x_k\}  = T_2 - \{x_1, \ldots, x_k\}$.
By the observation above, this implies that there is an LPR sequence $(x'_1, \dots, x'_k)$ turning $T_1$ into $T_2$ such that $\{x_1, \ldots, x_k\} = \{x'_1, \ldots, x'_k\}$ and $x'_1 = x$.
\qed
\end{proof}

\section{Proof of Theorem~\ref{thm:fpt-in-d}}

This section makes use of the concept of LPR moves, which are introduced in the previous section.
As discussed in the main paper, we employ a branch-and-bound style algorithm, in which at each step we alter a candidate solution by pruning and regrafting a leaf.
That is, we apply an LPR move.

The technically challenging part is bound the number of possible LPR moves to try. To do this, we will prove Lemma~\ref{lem:disagreementKernel2}, which provides a bound on the number of leaves to consider, and Lemma~\ref{lem:dont-check-too-many-trees}, which bounds the number of places a leaf may be regrafted to.


Denote by $tr(T)$ the set of rooted triplets of a tree $T$.
Two triplets $\triplet_1 \in tr(T_1)$ and $\triplet_2 \in tr(T_2)$ are \emph{conflicting}
if $\triplet_1 = ab|c$ and $\triplet_2 \in \{ac|b, bc|a\}$.
We denote by $conf(T_1, T_2)$ the set of triplets of $T_1$ for which 
there is a conflicting triplet in $T_2$.
That is, $conf(T_1, T_2) = \{ab|c \in tr(T_1) : ac|b \in tr(T_2)$ or $bc|a \in tr(T_2)\}$.
Finally we denote by $confset(T_1, T_2) = \{ \{a,b,c\} : ab|c \in conf(T_1, T_2) \}$,
\textit{i.e.} the collection of 3-label sets formed by conflicting triplets.
Given a collection $C = \{S_1, \ldots, S_{|C|}\}$ of sets, 
a \emph{hitting set} of $C$ is a set $S$ such that 
$S \cap S_i \neq \emptyset$ for each $S_i \in C$.

\begin{lemma}\label{lem:hit-triplets}
Let $X \subseteq \lblset$.  Then $T_1 - X =T_2 - X$ if and only if 
$X$ is a hitting set of $confset(T_1, T_2)$.
\end{lemma}

\begin{proof}
It is known that for two rooted trees $T_1, T_2$ that are leaf-labelled and binary,  $T_1 = T_2$ if and only if $tr(T_1) = tr(T_2)$~\cite{bryant1997building}.
Note also that $tr(T - X) = \{ab|c \in tr(T_1): X\cap \{a,b,c\}= \emptyset\}$ for any tree $T$ and $X\subseteq \lblset$.

Therefore we have that  $T_1 - X = T_2 - X$ if and only if  $tr(T_1 - X) = tr(T_2 - X)$, which holds if and only if for every $a,b,c \in \lblset \setminus X$, if $ab|c \in tr(T_1)$ then $ab|c \in tr(T_2)$. This in turn occurs if and only if $X$ is a hitting set for $confset(T_1, T_2)$.
\qed
\end{proof}

In what follows, we call $X \subseteq \lblset$ a \emph{minimal disagreement} between $T_1$ and $T_2$
if $T_1 - X = T_2 - X$ and for any $X' \subset X$, $T_1 - X' \neq T_2 - X'$.

\begin{lemma}\label{lem:must-move-x}
Suppose that $d < d_{LR}(T_1, T_2) \leq d' + d$ with $d' \leq d$, and that there is a tree $T^*$ and subsets $X_1, X_2 \subseteq \lblset$ such that 
$T_1 - X_1 = T^* - X_1$, $T_2 - X_2 = T^* - X_2$
and $|X_1| \leq d',|X_2| \leq d$.
Then, there is a minimal disagreement $X$ between $T_1$ and $T_2$ of size at most $d + d'$ and $x \in X$ such that $x \in X_1 \setminus X_2$.
\end{lemma}

\begin{proof}
Let $X' = X_1 \cup X_2$.
Observe that $T_1 - X' = T^* - X'= T_2 - X'$ and $|X'| \leq d+d'$.
 Letting $X$ be the minimal subset of $X'$ such that $T_1 - X = T_2 - X$, we have that $X$ is a minimal disagreement between $T_1$ and $T_2$ and $|X| \leq d+d'$.
 Furthermore as $|X| \geq  d_{LR}(T_1, T_2) > d$,
 $|X\setminus X_2| > 0$, and so there is some $x \in X$ with $x \in X\setminus X_2 = X_1 \setminus X_2$.
\qed
\end{proof}


We are now ready to state and prove Lemma~\ref{lem:disagreementKernel2}.




\begin{lemma}\label{lem:disagreementKernel2} 
%
Suppose that $d_{LR}(T_1, T_2) \leq d$ for some integer $d$.  Then, there is some $S \subseteq \lblset$ such that $|S| \leq 8d^2$, and for any minimal disagreement $X$ between $T_1$ and $T_2$ with $|X| \leq d$, $X \subseteq S$.
Moreover $S$ can be found in time $O(n^2)$.
\end{lemma}

We will call $S$ as described in Lemma~\ref{lem:disagreementKernel2} a
\emph{$d$-disagreement kernel} between $T_1$ and $T_2$.
Thus Lemma~\ref{lem:must-move-x} essentially states that if $T_1$ isn't a solution and $d_{LR}(T_1, T_2) > d$, then for $T_1$ to get closer to a solution, there is a leaf $x$ in the $d_{LR}(T_1, T_2)$-disagreement kernel that needs to be removed and regrafted in a location
that $T_2$ `agrees with'.
Lemma~\ref{lem:disagreementKernel2} in turn gives us a set $S$ of size at most $8d^2$ such  that the desired $x$ must be contained in $S$.

\begin{proof}
By Lemma~\ref{lem:hit-triplets}, it is enough to find a set $S$ such that $S$ contains every minimal hitting set of $confset(T_1,T_2)$ of size at most $d$.

We construct $S$ as follows.

Let $X$ be a subset of $\lblset$ of size at most $d$ such that $T_1 - X = T_2 - X$. 
As previously noted, this can found in time $O(n \log n)$~\cite{DBLP:journals/siamcomp/ColeFHPT00}.

For notational convenience, for each $x \in X$  we let $x_1,x_2$ be two new labels, and set $X_1 = \{x_1: x \in X\}$, $X_2 = \{x_2: x \in X\}$.
Thus, $X_1,X_2$ are disjoint ``copies'' of $X$.
Let $T_1'$ be derived from $T_1$ by replacing every label from $X$ with the corresponding label in $X_1$, and similarly
let $T_2'$ be derived from $T_2$ by replacing every label from $X$ with the corresponding label in $X_2$.

Let $T_J$ be a tree with label set $(\lblset \setminus X) \cup X_1 \cup X_2$ such that $T_J - X_2 = T_1'$ and $T_J -X_1 =T_2'$.
The tree $T_J$ always exists and can be found in polynomial time.  Intuitively, we can start from $T_1'$, and graft the leaves of $X_2$ where $T_2$ ``wants'' them to be.
See Figure~\ref{fig:joinTree} for an example.
Algorithm~\ref{alg:join-trees} gives a method for constructing $T_J$, and takes $O(n^2)$ time.  

\begin{figure}%
\centering
\subfloat{\includegraphics[width=\textwidth] {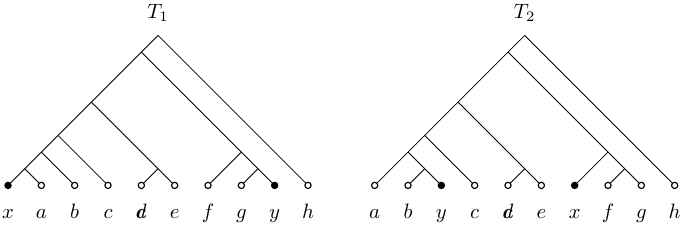}}\qquad
\subfloat{\includegraphics[width=0.5\textwidth] {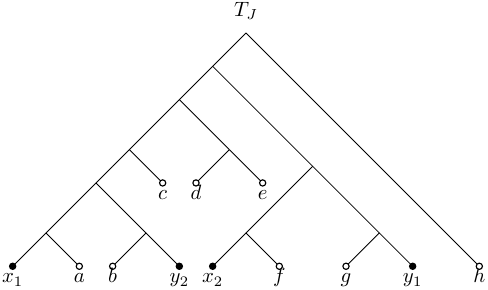}}\qquad
\caption{Construction of the tree $T_J$, given two trees $T_1, T_2$, with $X = \{x,y\}$ such that $T_1 - X = T_2 - X$.}\label{fig:LPRsequence}
\label{fig:joinTree}
\end{figure}

\begin{algorithm}
\caption{Algorithm to construct ``Join tree'' of $T_1',T_2'$}
\begin{algorithmic}[1]
\Procedure{join-trees}{$T_1', T_2',  L', X_1', X_2'$}
\Statex $T_1'$ is a tree on $L' \cup X_1'$, 
$T_2'$ is a tree on $L' \cup X_2'$,
$T_1'|_{L'} = T_2'|_{L'}$.
Output: A tree $T_J$ on $L' \cup X_1' \cup X_2'$
such that $T_J|_{L' \cup X_1'} = T_1'$
and $T_J|_{L' \cup X_2'} = T_2'$

\If{$L' \cup X_1'=\emptyset$}
  \State Return $T_2'$
\ElsIf{$L' \cup X_2'=\emptyset$}
  \State Return $T_1'$
\ElsIf{$X_1' \cup X_2' = \emptyset$}
  \State Return $T_1'$
\EndIf

\State Set $r_1 =$ root of $T_1'$, $u,v$ the children of $r_1$
\State Set $r_2 =$ root of $T_2'$, $w,z$ the children of $r_2$
\State Set $X_{1u} =$ descendants of $u$ in $X_1'$,  $L_{1u} =$ descendants of $u$ in $L'$
\State Set $X_{1v} =$ descendants of $v$ in $X_1'$,  $L_{1v} =$ descendants of $v$ in $L'$
\State Set $X_{2w} =$ descendants of $w$ in $X_2'$,  $L_{2w} =$ descendants of $w$ in $L'$
\State Set $X_{2z} =$ descendants of $z$ in $X_2'$,  $L_{2z} =$ descendants of $z$ in $L'$

\If{$L_{1u} = L_{2w}$ and $L_{1v} = L_{2z}$}
  \State Set $T_{left} = \textsc{join-trees}(T_1'|_{L_{1u} \cup X_{1u}}, T_2'|_{L_{1u} \cup X_{2w}}, L_{1u}, X_{1u}, X_{2w})$
  \State Set $T_{right} = \textsc{join-trees}(T_1'|_{L_{1v} \cup X_{1v}}, T_2'|_{L_{1v}\cup X_{2z}}, L_{1v}, X_{1v}, X_{2z})$
\ElsIf{$L_{1u} = L_{2z}$ and $L_{1v} = L_{2w}$}
  \State Set $T_{left} = \textsc{join-trees}(T_1'|_{L_{1u} \cup X_{1u}}, T_2'|_{L_{1u}\cup X_{2z}}, L_{1u}, X_{1u}, X_{2z})$
  \State Set $T_{right} = \textsc{join-trees}(T_1'|_{L_{1v} \cup X_{1v}}, T_2'|_{L_{1v} \cup X_{2w}}, L_{1v}, X_{1v}, X_{2w})$
\ElsIf{$L_{1u} = _\emptyset$}
  \State Set $T_{left} = \textsc{join-trees}(T_1'|_{X_{1u}}, T_2'|_\emptyset, \emptyset, X_{1u}, \emptyset)$
  \State Set $T_{right} = \textsc{join-trees}(T_1'|_{L' \cup X_{1v}}, T_2', L', X_{1v}, X_2')$
\ElsIf{$L_{1v} = \emptyset$}
  \State Set $T_{left} = \textsc{join-trees}(T_1'|_{L' \cup X_{1u}}, T_2', L', X_{1u}, X_2')$
  \State Set $T_{right} = \textsc{join-trees}(T_1'|_{X_{1v}}, T_2'|_\emptyset, \emptyset, X_{1v}, \emptyset)$
\ElsIf{$L_{2w} = \emptyset$}
  \State Set $T_{left} = \textsc{join-trees}(T_1'|_{\emptyset}, T_2'|_{X_{2w}}, \emptyset, \emptyset, X_{2w})$
  \State Set $T_{right} = \textsc{join-trees}(T_1', T_2'|_{L' \cup X_{2z}}, L', X_1', X_{2z})$
\ElsIf{$L_{2z} = \emptyset$}
  \State Set $T_{left} = \textsc{join-trees}(T_1', T_2'|_{L' \cup X_{2w}}, L', X_1', X_{2w})$
  \State Set $T_{right} = \textsc{join-trees}(T_1'|_{\emptyset}, T_2'|_{X_{2z}}, \emptyset, \emptyset, X_{2z})$
\EndIf
\Comment{If none of the above cases holds, then $T_1'|_{L'} \neq T_2'|_{L'}$, contradicting the requirements on the input}

\State Set $T_J =$ the tree on $L' \cup X_1' \cup X_2'$ whose root has $T_{left}$ and $T_{right}$ as children.
\State Return $T_J$.
\EndProcedure
\end{algorithmic}\label{alg:join-trees}
\end{algorithm}

In addition, let $L$ be the set of all labels in $\lblset\setminus X$ that are descended in $T_J$ from  $\lca_{T_J}(X_1 \cup X_2)$, and let $R = \lblset \setminus (L\cup X)$.
Thus, $L,X,R$ form a partition of $\lblset$, and $L,X_1,X_2,R$ form a partition of the labels of $T_J$.

For the rest of the proof, we call $\{x,y,z\}$ a \emph{conflict triple} if $\{x,y,z\} \in confset(T_1,T_2)$.

We first observe that no triple in $confset(T_1,T_2)$ contains a label in $R$. Indeed, consider a triple $\{x,y,z\}$. Any conflict triple must contain a label from $X$, so assume without loss of generality that $x \in X, z \in R$.
If $x \in X, y \in L, z \in R$, then we have that $T_J$ contains the triplets $x_1y|z, x_2y|z$, and so $T_1$ and $T_2$ both contain $xy|z$, and $\{x,y,z\}$ is not a conflict triple. Similarly if $x,y \in X, z \in R$, then $T_J$ contains the triplets $x_1y_1|z, x_2y_2|z$, and again $\{x,y,z\}$ is not a conflict triple.
If $x \in X$ and $y,z \in R$, then the triplet on $\{x_1,y,z\}$ in $T_J$ depends only on the relative positions in $T_J$ of $y,z$ and $\lca_{T_J}(X_1 \cup X_2)$. Thus we get the same triplet if we replace $x_1$ with $x_2$, and so $\{x,y,z\}$ is not a conflict triple.

This concludes the proof that no triple in $confset(T_1,T_2)$ contains a label in $R$. Having shown this, we may conclude that any minimal disagreement between $T_1$ and $T_2$ is disjoint from $R$, and so our returned set $S$ only needs to contain labels in $L \cup X$. 

Now consider the tree $T^* = T_J|_{X_1 \cup X_2}$, \textit{i.e.} the subtree of $T_J$ restricted to the labels in $X_1 \cup X_2$.
Thus in the example of Figure~\ref{fig:joinTree}, $T^*$ is the subtree of $T_J$ spanned by $\{x_1,x_2,y_1,y_2\}$.
We will now use the edges of $T^*$ to form a partition of $L$, as follows.
For any edge $uv$ in $T^*$ with $u$ the parent of $v$, let $s(uv)$ denote the set of labels $y \in \lblset$ such that $y$ has an ancestor  which is an internal node on the path from $u$ to $v$ in $T_J$, but $y$ is not a descendant of $v$ itself.
For example in Figure~\ref{fig:joinTree}, if $u$ is the least common ancestor of $x_1,y_1$ and $v$ is the least common ancestor of $x_1,y_2$, then $uv$ is an edge in $T^*$ and $s(uv) = \{c,d,e\}$.

Observe that $\{s(uv):uv \in E(T^*)\}$ forms a partition of $L$. 
(Indeed, for any $l \in L$, 
let $u$ be the minimal element in $T^*$ on the path in $T_J$ between $l$ and $\lca_{T_J}(X_1 \cup X_2)$ (note that $u$ exists as $\lca_{T_J}(X_1 \cup X_2)$ itself is in $T^*$). As $u$ is in $T^*$, both of its children are on paths in $T_J$ between $u$ and a child of $u$ in $T^*$. In particular, the child of $u$ which is an ancestor of $l$ is an internal node on the path between $u$ and $v$ in $T_J$, for some child $v$ of $u$ in $T^*$, and $l$ is not descended from $v$ by construction.
It is clear by construction that all $s(uv)$ are disjoint.)

The main idea behind the construction of $S$ is that we will add $X$ to $S$, together with $O(d)$ labels from $s(uv)$ for each edge $uv$ in $T^*$.
As the number of edges in $T^*$ is $2(|X_1 \cup X_2| -1) = O(d)$, we have the required bound of $O(d^2)$ on $|S|$. 

So now consider $s(uv)$ for some edge $uv$ in $T^*$.
In order to decide which labels to add to $S$, we need to further partition $s(uv)$.
Let $u = u_0u_1\dots u_t = v$ be the path in $T_J$ from $u$ to $v$.
For each $i \in [t-1]$ (note that this does not include $i=0$), we call the set of labels descended from $u_i$ but not $u_{i+1}$ a \emph{dangling clade}.
Observe that the dangling clades form a partition of $s(uv)$.
Thus in the example of Figure~\ref{fig:joinTree}, if $u$ is the least common ancestor of $x_1,y_1$ and $v$ is the least common ancestor of $x_1,y_2$, then for the edge $uv$ the dangling clades are $\{c\}$ and $\{d,e\}$f.

We now make the following observations about the relation between $s(uv)$ and triples in $confset(T_1,T_2)$. 

{\bf Observation 1:}
if $\{x,y,z\}$ is a conflict triple and $x \in s(uv), y,z \notin s(uv)$, then $\{x',y,z\}$ is also a conflict triple for any $x' \in s(uv)$.
(The intuition behind this is that there are no labels appearing 'between' $x$ and $x'$ that are not in $s(uv)$.)

{\bf Observation 2:}
for any triple $\{x,y,z\}$ with $x,y \in s(uv)$, $\{x,y,z\}$ is a conflict triple if and only if $x,y$ are in different dangling clades and $z \in X$ with $z_i$ descended from $v$, $z_{3-i}$ not descended from $u_1$ for some $i \in [2]$ (recall that $z_1 \in X_1$ and $z_2 \in X_2$). 
To prove one direction, it is easy to see that if the conditions hold, then $T_i$ displays either $xz|y$ or $yz|x$ (depending on which dangling clade appears 'higher'), and $T_{3-i}$ displays $xy|z$.
For the converse, observe first that $z\in X$ as $X$ is a hitting set for $confset(T_1,T_2)$ and $x,y \notin X$. Then if $xy$ are in the same dangling clade, we have that both $T_1$ and $T_2$ display $xy|z$. So $x,y$ must be in different dangling clades.
Next observe that each of $z_1,z_2$ must either be  descended from $v$ or not descended from $u_1$, as otherwise $v$ would not be the child of $u$ in $T^*$.
If $z_1,z_2$ are both descended from $v$ or neither are descended from $u_1$, then $T_1$ and $T_2$ display the same triplet on $\{x,y,z\}$. 
So instead one must be descended from $v$ and one not descended from $u_1$, as required.

Using Observations 1 and 2, we now prove the following:

{\bf Observation 3:}
for any minimal disagreement $X'$ between $T_1$ and $T_2$, one of the following holds:
\begin{itemize}
 \item $X' \cap s(uv) = \emptyset$;
 \item $s(uv) \subseteq X'$;
 \item $s(uv) \setminus X'$ forms a single dangling clade.
\end{itemize}

To see this, let $X'$ be any minimal hitting set of $confset(T_1,T_2)$ with $s(uv) \cap X' \neq \emptyset$ and $s(uv) \setminus X' \neq \emptyset$.
As $X'$ is minimal, any $x \in s(uv) \cap X'$ must be in a conflict triple $\{x,y,z\}$ with $y,z \notin X'$. 
As $X$ is a hitting set for $confset(T_1,T_2)$, at least one of $y,z$ must be in $X$. 
If $y,z \notin s(uv)$, then by Observation 1 $\{x',y,z\}$ is also a conflict triple for any $x' \in s(uv) \setminus X'$. But this is a contradiction as $\{x',y,z\}$ has no elements in $X'$. Then one of $y,z$ must also be in $s(uv)$. Suppose without loss of generality that $y \in s(uv)$. We must also have that $z \in X$, as $X$ is a hitting set for $confset(T_1,T_2)$ and $x,y \notin X$. By Observation 2, we must have that one of $z_1,z_2$ is descended from $v$, and the other is not descended from $u_1$.
This in turn implies (again by Observation 2) that for any $x' \in s(uv) \setminus X'$, if $x'$ and $y$ are in different dangling clades then $\{x',y,z\}$ is a conflict triple. Again this is a contradiction as $\{x',y,z\}$ has no elements of $X'$, and so we may assume that all elements of $s(uv) \setminus X'$ are in the same dangling clade.

It remains to show that every element of this dangling clade is in $s(uv) \setminus X'$.
To see this, suppose there exists some $x \in X'$ in the same dangling clade as the elements of $s(uv) \setminus X'$.
Once again we have that $x$ is in some conflict triple $\{x,y,z\}$ with $y,z \notin X'$, and if $y,z \notin s(uv)$ then $\{x',y,z\}$ is also a conflict triple for any $x' \in s(uv) \setminus X'$, a contradiction.
So we may assume that one of $y,z$ is in $s(uv) \setminus X'$. 
But all elements of $s(uv) \setminus X'$ are in the same dangling clade as $x$, and so by Observation 2 $\{x,y,z\}$  cannot be a conflict triple, a contradiction.
So finally we have that all elements  of $s(uv) \setminus X'$ are in the same dangling clade and all elements of this clade are in $s(uv) \setminus X'$, as required.

With the proof of Observation 3 complete, we are now in a position to construct $S$.
For any minimal hitting set $X'$ of $confset(T_1,T_2)$ with size at most $d$, by Observation 3 either 
$X' \cap s(uv) = \emptyset$, or $s(uv) \subseteq X'$ (in which case $|s(uv)|\leq d$), or $s(uv) \setminus X'$ forms a single dangling clade $C$ (in which case  $|s(uv) \setminus C| \leq d$).

So add all elements of $X$ to $S$.
For all $uv \in E(T_J)$ and any dangling clade $C$ of labels in $s(uv)$, add $s(uv) \setminus C$ to $S$ if $|s(uv) \setminus C| \leq d$.
Observe that this construction adds at most $2d$ labels from $s(uv)$ to $S$.

Thus, in total, we have that the size of $S$ is at most $|X| + 2d|E(T_J)| \leq d+2d(2(|X_1 \cup X_2|-1)) \leq d + 2d(4d-2) = 8d^2-3d \leq 8d^2$.  

Algorithm~\ref{alg:disagreementKernel2} describes the full procedure formally. 
The construction of $T_J$ occurs once and as noted above takes $O(n^2)$ time.
As each other line in the algorithm is called at most $n$ times and takes $O(n)$ time, the overall running time of the algorithm $O(n^2)$.
\qed
\end{proof}

\begin{algorithm}
\caption{Algorithm to construct a $d$-disagreement kernel between $T_1$ and $T_2$}
\begin{algorithmic}[1]
\Procedure{disagreement-kernel}{$d, T_1, T_2$}
\Statex $T_1$ and $T_2$ are trees on $\lblset$,$d$ an integer.

Output: A set $S \subseteq \lblset$ such that for every minimal disagreement $X$ between $T_1$ and $T_2$ with $|X| \leq d$, $X \subseteq S$.

\State Find $X$  such that $|X| \leq d$ and $T_1  - X  = T_2 - X$
\State Set $S = X$
\State Let $X_1, X_2$ be copies of $X$ and replace $T_1,T_2$ with corresponding trees $T_1',T_2'$ on $(\lblset \setminus X)\cup X_1, (\lblset \setminus X)\cup X_2$.
\State Let $T_J = \textsc{join-trees}(T_1', T_2', (\lblset \setminus X), X_1, X_2)$
\State Let $T^*= T_J|_{X_1 \cup X_2}$
\For{$uv \in E(T^*)$}
  \State Let $u = u_0u_1 \dots u_t = v$ be the path in $T_J$ from $u$ to $v$
  \State Let $s(uv) = \{l \in \lblset \setminus X: l$ is descended from $u_1$ but not from $v\}$
  \State Set $p = |s(uv)| - d$
  \Comment Any clade $C$ has $|C|\geq p$ iff $|s(uv) \setminus C| \leq d$
  \For{$i \in [t]$}
    \State Set $C = \{l \in s(uv): l$ is descended from $u_i$ but not from $u_{i+1} \}$
    \Comment{$C$ is a single 'dangling clade'}
    \If{$|C|\geq p$}
      \State Set $S = S \cup (s(uv) \setminus C)$
    \EndIf
  \EndFor
\EndFor
\State Return $S$.
\EndProcedure
\end{algorithmic}\label{alg:disagreementKernel2}
\end{algorithm}

The last ingredient needed for Theorem~\ref{thm:fpt-in-d} is Lemma~\ref{lem:dont-check-too-many-trees}, 
which shows that if a leaf $x$ of $T_1$ as described in Lemma~\ref{lem:must-move-x} has to be moved, then there are not too many 
ways to regraft it in order to get closer to $T^*$.

In the course of the following proofs, we will want to take observations about one tree and use them to make statements about another.
For this reason it's useful to have a concept of one node "corresponding" to another node in a different tree. 
In the case of leaf nodes this concept is clear - two leaf nodes are equivalent if they are assigned the same label- but for internal nodes there is not necessarily any such correspondence. 
However, in the case that one tree is the restriction of another to some label set, we can introduce a well-defined notion of correspondence:

Given two trees $T, T'$ such that $T' = T|_X$ for some $X \subseteq \lblset(T)$, and a node $u \in V(T')$, define the node $u_T$ of $T$ by $u_T = \lca_T(\L_{T'}(u))$.
That is, $u_T$ is the least common ancestor, in $T$, of the set of labels belonging to descendants of $u$ in $T'$. 
We call $u_T$ the \emph{node corresponding to $u$ in $T$}.

We note two useful properties of $u_T$ here: 

\begin{lemma}\label{lem:corrAncestor}
For any $T, T', X \subseteq \lblset(T)$ such that $T' = T|_X$ and any $u,v \in V(T')$, $u_T$ is an ancestor of $v_T$ if and only if $u$ is an ancestor of $v$.
\end{lemma}
\begin{proof}
If $u$ is an ancestor of $v$ then
$\L_{T'}(v) \subseteq \L_{T'}(u)$, which implies that $u_T$ is an ancestor of  $v_T$.
For the converse, observe that 
for any $Z \subseteq X$,
any label in $X$ descending from $\lca_T(Z)$ in $T$ is also descending from $\lca_{T'}(Z)$ in $T'$.
In particular letting $Z = \L_{T'}(u)$, we have $\L_{T}(u_T)\cap X  = \L_{T}(\lca_T(Z)) \cap X \subseteq \L_{T'}(\lca_{T'}(Z)) =  \L_{T'}(\lca_{T'}(\L_{T'}(u))) = \L_{T'}(u) \subseteq \L_{T}(u_T)\cap X$.
Thus $\L_{T'}(u) = \L_T(u_T) \cap X$ and similarly $\L_{T'}(v) = \L_T(v_T) \cap X$.
Then we have that $u_T$ being an ancestor of $v_T$ implies $\L_T(v_T) \subseteq L_T(u_T)$, which implies that $\L_{T'}(v) = \L_T(v_T) \cap X \subseteq  L_T(u_T) \cap X = \L_{T'}(u)$, which implies that $u$ is an ancestor of $v$.
\qed
\end{proof}

\begin{lemma}\label{lem:corrTransitive}
 For any $T'', T', T$ and  $Y \subseteq X \subseteq \lblset(T)$ such that $T' = T|_{X}$ and  $T'' = T'|_{Y}$, $(u_{T'})_{T} = u_{T}$.  
\end{lemma}
\begin{proof}
It is sufficient to show that any node in $V(T)$ is a common ancestor of $\L_{T'}(\lca_{T'}(Z))$  if and only if it is a  common ancestor of $Z$, where $Z = \L_{T''}(u)$ (as this implies that the least common ancestors of these two sets are the same).
It is clear that if $v \in V(T)$ is a common ancestor of $\L_{T'}(\lca_{T'}(Z))$ then it is also a common ancestor of $Z$, as $Z \subseteq  \L_{T'}(\lca_{T'}(Z))$.
For the converse, observe that 
as $T' = T|_X$  and $Z\subseteq X$, any label in $X$ descended from $\lca_{T'}(Z)$ in $T'$  is also descended from $\lca_T(Z)$ in $T$.
This implies $\L_{T'}(\lca_{T'}(Z)) \subseteq \L_T(\lca_T(Z))$, and so any common ancestor of $Z$ in $T$ is also a common ancestor of $\L_{T'}(\lca_{T'}(Z))$.
\qed
\end{proof}

We are now ready to state and prove Lemma~\ref{lem:dont-check-too-many-trees}

\begin{lemma}\label{lem:dont-check-too-many-trees}
Suppose that $d < d_{LR}(T_1, T_2) \leq d' + d$ with $d' \leq d$, and that there are  $X_1, X_2 \subseteq \lblset$, and a tree $T^*$ such that $T_1 - X_1 = T^* - X_1, T_2 - X_2 = T^* - X_2, |X_1| \leq d', |X_2| \leq d$, and let $x \in X_1 \setminus X_2$.
Then, there is a set $P$ of trees on label set $\lblset$ that satisfies the following conditions:

\begin{itemize}
\item
for any tree $T'$ such that $d_{LR}(T', T^*) < d_{LR}(T_1,T^*)$ and $T'$ can be obtained 
  from $T_1$ by pruning a leaf $x$ and regrafting it, $T' \in P$;
  
\item
$|P| \leq 18(d+d')+8$;

\item
$P$ can be found in time $O(n(\log n + 18(d+d')+8))$.

\end{itemize}
\end{lemma}

The idea behind the proof is as follows:
by looking at a subtree common to $T_1$ and $T_2$, we can identify the location that $T_2$ ``wants" $x$ to be positioned. 
This may not be the correct position for $x$, but we can show that if $x$ is moved too far from this position,
we will create a large number of conflicting triplets between $T_2$ and the solution $T^*$.
As a result, we can create all trees in $P$ by removing $x$ from $T_1$ and grafting it on one of a limited number of edges.

\begin{proof}
For the purposes of this proof, we will treat each tree $T$ as ``planted'', \textit{i.e.} as having an additional root of degree $1$, denoted $r(T)$, as the parent of what would normally be considered the ``root" of the tree. (That is, $r(T)$ is the parent of $\lca_T(\lblset(T))$. Note that trees are otherwise binary. We introduce $r(T)$ as a notational convenience to avoid tedious repetition of proofs - grafting a label above a tree $T$ can instead be represented as grafting it on the edge between $r(T)$ and its child.
For the purposes of corresponding nodes, if $T' = T - X$ then $(r(T'))_T = r(T)$. 
This  allows us to assume that every node in $T$ is a descendant of $u_T$ for some node $u$ in $T'$.

%
A naive method for constructing a tree in $P$ is the following:
Apply an LPR move $(x,e)$ on $T_1$, such that $x$ is moved to a position that $T_2$ 
``wants'' $x$ to be in.
There are at least two problems with this method.
The first is that, since $T_1$ and $T_2$ have different structures, it is not clear where in $T_1$ it is that $T_2$ ``wants'' $x$ to be. We can partially overcome this obstacle by initially considering a subtree common to both $T_1$ and $T_2$. However, because $T_2$ will want to move leafs that will not be moved in $T_1$, it can still be the case that even though $T_2$ ``agrees'' with $T^*$ on $x$, $T_2$ may want to put $x$ in the ``wrong'' place, when viewed from the perspective of $T_1$. For this reason we have to give a counting argument to show that if $x$ is moved ``too far'' from the position suggested by $T_2$, it will create too many conflicting triplets, which cannot be covered except by moving $x$. We make these ideas precise below.

\medskip

Let $P^*$ be the set of all trees $T'$ such that $d_{LR}(T', T^*) < d_{LR}(T_1, T^*) $ and $T'$ can be obtained 
  from $T_1$ by an LPR move on $x$.
  Thus, it is sufficient to construct a set $P$ such that $|P| \leq 18(d + d') + 8$ and $P^* \subseteq P$.

We first construct a set $X_m \subseteq \lblset$ such that $|X_m| \leq d+ d', x \in X_m$, and $T_1 - X_m= T_2 - X_m$.
Note that the unknown set $(X_1 \cup X_2)$ satisfies these properties, as $T_1 - (X_1 \cup X_2) = T^* - (X_1 \cup X_2) =  T_2 - (X_1 \cup X_2)$, and so such a set $X_m$ must exist.
We can find $X_m$ in 
time $O(n \log n)$
by applying MAST on $(T_1 - \{x\}, T_2 - \{x\})$~\cite{DBLP:journals/siamcomp/ColeFHPT00}.

Now let $T_m$ be the tree with labelset $\lblset\setminus {X_m}$ such that $T_m = T_1 - X_m = T_2 - X_m$. Note that for any $T'$ in $P^*$, we have that $T' - \{x\} = T_1 - \{x\}$ and therefore $T' - X_m = T_1 - X_m = T_m$.

Informally, we now have a clear notion of where $T_2$ ``wants'' $x$ to go, relative to $T_m$. There is a unique edge $e$ in $T_m$ such that grafting $x$ on $e$ will give the tree $T_2 - (X_m\setminus\{x\})$. If we assume that this is the ``correct'' position to add $x$, then it only remains to add the remaining labels of $X_m$ back in a way that agrees with $T_1$ (we will describe how this can be done at the end of the proof).
Unfortunately, grafting $x$ onto the obvious choice $e$ does not necessarily lead to a graph in $P^*$. This is due to the fact that $T_2$ can be ``mistaken'' about labels outside of $X_m$.

To address this, we have 
try grafting $x$ on other edges of $T_m$.
There are too many edges to try them all. We therefore need the following claim, which allows us to limit the number of edges to try.

\medskip
{\bf Claim:} \emph{In $O(n)$ time, we can find $y \in V(T_m)$ and $Z\subseteq V(T_m), |Z| \leq 4$, such that:
\begin{itemize}
  \item For any $T'$ in $P^*$, $x \in \L_{T'}(y_{T'}) \setminus \bigcup_{z'\in Z}\L_{T'}(z'_{T'})$
  \item $|\L_{T_m}(y) \setminus \bigcup_{z'\in Z}\L_{T_m}(z')| \leq 8(d+d')$
\end{itemize}
}

Informally, the claim identifies a node $y$ and set of nodes $Z$ in $T_m$, such that $x$ should be added as a descendant of $y$ but not of any node in $Z$, and the number of such positions is bounded. 
Algorithm~\ref{alg:locationRestrictions} describes the formal procedure to produce $y$ and $Z$.
The proof of the claim takes up most of the remainder of our proof; the reader may wish to skip it on their first readthrough.

\begin{algorithm}
\caption{FPT algorithm to restrict possible locations of $x$ given $(T_m,T_1,T_2,x, d, d')$}
\begin{algorithmic}[1]
\Procedure{location-restriction}{$T_m, T_2, X_m, x, d, d'$}
\Statex $T_1, T_2$ are two trees, $T_m$ is a common subtree of $T_1$ and $T_2$ such that $T_m = T_2 - X_m$, $x$ is a label that cannot be moved in $T_2$ (but must be moved in $T_1$), $d$ is the maximum number of leaves we can remove in a tree, $d'$ is the maximum number of leaves we can move in $T_1$. Output is a pair $(y,Z)$  with $y \in V(T_m)$, $Z\subseteq V(T_m)$, such that we may assume $x$ is a descendant of $y$ but not a descendant of any $z' \in Z$, and the number of labels like this in $T_m$ is $O(d)$.
For this pseudocode, every tree $T$ has a degree-$1$ root $r(T)$.

\State Set $T_m' = T_2 - (X_m \setminus \{x\})$
\State Set $z = $ lowest ancestor of $x$ in $T_m'$ such that $|\L_{T_m'}(z)\setminus{x}| \geq d+d'$, or return $(r(T_m'), \emptyset)$ if no such $z$ exists.
\State Set $y = $ lowest ancestor of $z$ in $T_m'$ such that $|\L_{T_m'}(y) \setminus \L_{T_m'}(z)| \geq d+d'$, or $r(T_m')$ if no such ancestor exists.
\State \Comment{Find sets $Z = Z_1 \cup Z_2$ of nodes that cover all but a bounded number of the descendants of $y$, and such that we can rule out $x$ being descended from any $z'$ in $Z$.}
\State Let $z_1,z_2$ be the children of $z$ such that $x$ is descended from $z_1$ in $T_m'$
\State Set $Z_1 = \{ z'$ descended from $z_2: |\L_{T_m'}(z')| \ge d+d'$ and $|\L_{T_m'}(z_2) \setminus \L_{T_m'}(z')| \ge d+d'$, and this does not hold for any ancestor of $z'\}$
\State Let $y_1,y_2$ be the children of $y$ such that $x$ is descended from $y_1$ in $T_m'$
\State Set $Z_2 = \{ y'$ descended from $y_2: |\L_{T_m'}(y')| \ge d+d'$ and $|\L_{T_m'}(y_2) \setminus \L_{T_m'}(y')| \ge d+d'$, and this does not hold for any ancestor of $y'\}$
\State \Comment{Note that $|Z_1| \leq 2, |Z_2| \leq 2$.}
\State Set $y^* = $ node of $T_m$ for which $y$ is the corresponding node in $T_m'$
\State Set $Z = \{z^*$ in $T_m : z'\in Z_1\cup Z_2$ is the node corresponding to $z^*$ in $T_m\}$
\State Return $(y^*,Z)$ 
\EndProcedure
\end{algorithmic}\label{alg:locationRestrictions}
\end{algorithm}

  \begin{proof}
   Let $T_m' = T_2 - (X_m\setminus  \{x\})$. Note that $T_m' - \{x\} = T_m$.
   We will use the presence of $x$ in $T_m'$ to identify the node $y$ and set $Z$. (Technically, this means the nodes we find are nodes in $T_m'$ rather than $T_m$. However, we note that apart the parent of $x$ and $x$ itself, neither of which will be added to $\{y\}\cup Z$, every node in $T_m'$ is the node $v_{T_m'}$ corresponding to some node $v$ in $T_m$. For the sake of clarity, we ignore the distinction and write $v$ to mean $v_{T_m'}$ throughout this proof. The nodes in $\{y\}\cup Z$ should ultimately be replaced with the nodes in $T_m$ to which they correspond.)
   
   We first identify two nodes $z,y$ of $T_m$ as follows:
    \begin{itemize}
     \item Let $z$ be the least ancestor of $x$ in $T_m'$ such that $|\L_{T_m'}(z) \setminus \{x\}| \ge d+d'$.
     If no such $x$ exists, then $\lblset(T_m') \leq d+d'$ and we may return $y = r(T_m'), Z = \emptyset$.
     \item Let $y$ be the least ancestor of $z$ in $T_m'$ such that $|\L_{T_m'}(y) \setminus \L_{T_m'}(z)| \ge d+d'$.
     If no such $y$ exists, set $y = r(T_m')$.
    \end{itemize}
    
    Using this definition, we will show that 
    $x$ must be a descendant of $y_{T'}$ for any $T' \in P^*$. We first describe a general tactic for restricting the position of $x$ in $T'$, as this tactic will be used a number of times.
    
    Suppose that for some $T' \in P^*$ there is a set of  $d+d'$ triplets in $confset(T',T_2)$ whose only common element is $x$.
    Then let $X' \subseteq \lblset$ be a set of  labels such that $T' - X' = T^* - X'$ and $|X'| = d_{LR}(T', T^*) \leq d_{LR}(T_1, T^*) - 1 \leq d' - 1$.
    Note that $T_2 - (X' \cup X_2) = T^* - (X' \cup X_2) = T' - (X' \cup X_2)$, and therefore $(X' \cup X_2)$ is a hitting set for $confset(T',T_2)$. 
    As $|X' \cup X_2| \leq d + d-1$ and there are $d+d'$ triplets in $confset(T',T_2)$ whose only common element is $x$, it must be the case that $x \in X'\cup X_2$. As $x\notin X_2$, we must have $x\in X'$. But this implies that $T_1 - X' = T' - X' = T^* - X'$ and therefore $d_{LR}(T_1, T^*) \leq |X'| = d_{LR}(T', T^*) \leq  d_{LR}(T_1, T^*) - 1$, a contradiction. Thus we may assume that such a set of triplets does not exist.
    
    We now use this idea to show that $x \in \L_{T'}(y_{T'})$, for any $T'\in P^*$.
    Indeed, suppose $x \notin \L_{T'}(y_{T'})$.
    We may assume $y\neq r(T_m')$ as otherwise $y_{T'}= r(T')$ by definition and so $\L_{T'}(y_{T'}) = \lblset(T')$. 
    Then let $z_1, \ldots, z_{d+d'}$ be $d+d'$ labels in $\L_{T_m'}(z)\setminus \{x\}$.
    Let $y_1, \ldots, y_{d+d'}$ be $d+d'$ labels in $\L_{T_m'}(y)\setminus \L_{T_m'}(z)$.    
    Observe that for each $i \in [d+d']$, $T_m'$ (and therefore $T_2$) contains the triplet $(z_ix|y_i)$, but $T'$ contains the triplet $(z_iy_i|x)$.
    Therefore  $confset(T',T_2)$ contains $d+d'$ sets whose only common element is $x$. 
    As this implies a contradiction, we must have  $x \in \L_{T'}(y_{T'})$.
    
    Note however that $|\L_{T_m'}(y)|$ maybe be very large.
    In order to provide a bounded range of possible positions for $x$, we still need to find a set $Z$ of nodes such that $|\L_{T_m'}(y) \setminus \bigcup_{z' \in Z} \L_{T_m'}(z'))|$ is bounded, and such that we can show $x \notin \L_{T'}(z'_{T'})$ for any $z' \in Z$.


%
    
    We now construct a set $Z_1$ of descendants of $z$ as follows:
    
    \begin{itemize}
     \item Let $z_1, z_2$ be the children of $z$ in $T_m'$ such that $x$ is descended from $z_1$. 
     \item If $|\L_{T_m'}(z_2)| \leq 3(d+d')$ then set $Z_1 = \emptyset$.
     \item Otherwise, let $Z_1$ be the set of highest descendants $z'$ of $z_2$, such that $|\L_{T_m'}(z')|\ge d+d'$ and $|\L_{T_m'}(z_2) \setminus \L_{T_m'}(z')|\ge d+d'$ (\textit{i.e.} by highest descendant we mean 
     such that $z'$ has no ancestor $z''$ with the same properties).
    \end{itemize}
    
    Note that $|\L_{T_m'}(z_1)| \leq d+d'$ by our choice of $z$.
    It follows that if $|\L_{T_m'}(z_2)| \leq 3(d+d')$ then $|\L_{T_m'}(z)| \le 4(d+d')$.
    If on the other hand $|\L_{T_m'}(z_2)| > 3(d+d')$ then  $Z_1$ is non-empty. Indeed, let $z'$ be a lowest descendant of $z_2$ with $|\L_{T_m'}(z')|\ge d+d'$, and observe that $|\L_{T_m'}(z')|\leq 2(d+d')$. Then either $z' \in Z_1$, or $|\L_{T_m'}(z_2) \setminus \L_{T_m'}(z')|\le d+d'$, in which case $|\L_{T_m'}(z_2)| \leq  |\L_{T_m'}(z_2) \setminus \L_{T_m'}(z')| + |\L_{T_m'}(z')| \leq d+d' + 2(d+d') = 3(d+d')$.
     
     We also have that $|Z_1| \leq 2$. Indeed, let $z_1',z_2',z_3'$ be three distinct nodes in $Z_1$, and suppose without loss of generality that $(z_1' z_2' | z_3') \in tr(T_m')$. Then setting $z' = \lca_{T_m'}(z_1',z_2')$, we have that $z'$ is an ancestor of $z_1'$ such that  $|\L_{T_m'}(z')|\ge d+d'$ and $|\L_{T_m'}(z_2) \setminus \L_{T_m'}(z')|\ge |\L_{T_m'}(z_3')| \ge d+d'$, a contradiction by minimality of $z_1$.
     
     We have that $|\L_{T_m'}(z) \setminus \bigcup_{z' \in Z_1} \L_{T_m'}(z'))| \leq 4(d+d')$. Indeed, if $Z_1 = \emptyset$ then $|\L_{T_m'}(z)| \leq 4(d+d')$ as described above. Otherwise, let $z'$ be an element of $Z_1$ and $z_p$ its parent, $z_s$ its sibling in $T_m'$. Clearly  $|\L_{T_m'}(z_p)| \geq |\L_{T_m'}(z')| \geq  d+d'$, and so as $z_p \notin Z_1$ we have $|\L_{T_m'}(z_2) \setminus \L_{T_m'}(z_p)| < d+d'$. 
     If $|\L_{T_m'}(z_s)| \geq  d+d'$ then $z_s \in Z_1$ (since $|\L_{T_m'}(z_2) \setminus \L_{T_m'}(z_s)| \geq |\L_{T_m'}(z')| \geq  d+d'$), and so  $|\L_{T_m'}(z) \setminus \bigcup_{z' \in Z_1} \L_{T_m'}(z'))| \leq |\L_{T_m'}(z_1) | +  |\L_{T_m'}(z_2) \setminus \L_{T_m'}(z_p)| \leq 2(d+d')$. 
     Otherwise, $|\L_{T_m'}(z) \setminus \bigcup_{z' \in Z_1} \L_{T_m'}(z'))| \leq |\L_{T_m'}(z_1) | +  |\L_{T_m'}(z_2) \setminus \L_{T_m'}(z_p)| +  |\L_{T_m'}(z_s)| \leq 3(d+d')$. 
     
    We have now shown that $|Z_1|\le 2$ and that  $|\L_{T_m'}(z) \setminus \bigcup_{z' \in Z_1} \L_{T_m'}(z'))| \leq 4(d+d')$. The final property of $Z_1$ we wish to show is that for any $z'\in Z_1$ and any $T' \in P$, $x \notin \L_{T'}(z'_{T'})$.

     So suppose  $x \in \L_{T'}(z'_{T'})$.
       Let $\hat{z}_1, \dots,  \hat{z}_{d+d'}$ be $d+d'$ labels in $\L_{T_m'}(z_2)\setminus L_{T_m'}(z')$. Also, $z_1$ and $z_2$ were already taken.
       Let $w_1, \dots, w_{d+d'}$ be $d+d'$ labels in $\L_{T_m'}(z')$.
       Then for each $i \in [d+d']$, $T_m'$ (and therefore $T_2$) contains the triplet $(\hat{z}_iw_i|x)$, but $T'$ contains the triplet $(xw_i|\hat{z}_i)$.
      Therefore  $confset(T',T_2)$ contains $d+d'$ sets whose only common element is $x$. 
      As this implies a contradiction, we must have   $x \notin \L_{T'}(z'_{T'})$.
    
    \medskip
    
    We now define a set $Z_2$ of descendants of $y$:
    
        \begin{itemize}
     \item If $y = r(T_m')$, set $Z_2 = \emptyset$.
     \item Otherwise, let $y_1, y_2$ be the children of $y$ in $T_m'$ such that $z$ is descended from $y_1$. 
     \item If $|\L_{T_m'}(y_2)| \leq 3(d+d')$ then  set $Z_2 = \emptyset$.
     \item Otherwise, let $Z_2$ be the set of highest descendants $y'$ of $y_2$, such that $|\L_{T_m'}(y')|\ge d+d'$ and $|\L_{T_m'}(y_2) \setminus \L_{T_m'}(y')|\ge d+d'$ (\textit{i.e.} such that $y'$ has no ancestor $y''$ with the same properties).
    \end{itemize}
     
     In a similar way to the proofs for $Z_1$, we can show that $|Z_2|\leq 2$, that $|(\L_{T_m'}(y)  \setminus \L_{T_m'}(z)) \setminus \bigcup_{y' \in Z_2} \L_{T_m'}(y'))| \leq 4(d+d')$, and that $x \notin \L_{T'}(y'_{T'})$ for any $y' \in Z_2$ and any $T' \in P^*$.
     
     
     Note that $|\L_{T_m'}(y_1)\setminus \L_{T_m'}(z) | \leq d+d'$ by our choice of $y$.
     It follows that if $|\L_{T_m'}(y_2)| \leq 3(d+d')$ then $|\L_{T_m'}(y) \setminus \L_{T_m'}(z)| \le 4(d+d')$.
     If on the other hand $|\L_{T_m'}(y_2)| > 3(d+d')$, then  $Z_2$ is non-empty. Indeed, let $y'$ be a lowest descendant of $y_2$ with $|\L_{T_m'}(y')|\ge d+d'$, and observe that $|\L_{T_m'}(y')|\leq 2(d+d')$. Then either $y' \in Z_2$, or $|\L_{T_m'}(y_2) \setminus \L_{T_m'}(y')|\le d+d'$, in which case $|\L_{T_m'}(y_2)| \leq  |\L_{T_m'}(y_2) \setminus \L_{T_m'}(y')| + |\L_{T_m'}(y')| \leq d+d' + 2(d+d') = 3(d+d')$.
     
     We also have that $|Z_2| \leq 2$. 
     Indeed, let $y_1',y_2',y_3'$ be three distinct nodes in $Z_2$, and suppose without loss of generality that $(y_1' y_2' | y_3') \in tr(T_m')$. Then setting $y' = \lca_{T_m'}(y_1',y_2')$, we have that $y'$ is an ancestor of $y_1'$ such that  $|\L_{T_m'}(y')|\ge d+d'$ and $|\L_{T_m'}(y_2) \setminus \L_{T_m'}(y')|\ge |\L_{T_m'}(y_3')| \ge d+d'$, a contradiction by minimality of $y_1$.
     
     We have that $|(\L_{T_m'}(y)  \setminus \L_{T_m'}(z)) \setminus \bigcup_{y' \in Z_2} \L_{T_m'}(y'))| \leq 4(d+d')$. Indeed, if
      $y = r(T_m')$ then  by construction $|\L_{T_m'}(\hat{y}) \setminus \L_{T_m'}(z)| < d+d'$ for any ancestor $\hat{y}$ of $z$ (noting that otherwise there would be no reason to set $y$ as $r(T_m')$ rather than the child of $r(T_m')$), and so in particular $|\L_{T_m'}(y) \setminus \L_{T_m'}(z)| < d+d'$. 
     If $y\neq r(T_m')$ and
     $Z_2 = \emptyset$ then $|\L_{T_m'}(y) \setminus \L_{T_m'}(z)| \leq 4(d+d')$ as described above. Otherwise, let $y'$ be an element of $Z_2$ and $y_p$ its parent, $y_s$ its sibling in $T_m'$. Clearly  $|\L_{T_m'}(y_p)| \geq |\L_{T_m'}(y')| \geq  d+d'$, and so as $y_p \notin Z_2$ we have $|\L_{T_m'}(y_2) \setminus \L_{T_m'}(y_p)| < d+d'$. 
     If $|\L_{T_m'}(y_s)| \geq  d+d'$ then $y_s \in Z_2$ (since $|\L_{T_m'}(y_2) \setminus \L_{T_m'}(y_s)| \geq |\L_{T_m'}(y')| \geq  d+d'$), and so  $|(\L_{T_m'}(y)  \setminus \L_{T_m'}(z))  \setminus \bigcup_{y' \in Z_2} \L_{T_m'}(y'))| \leq |\L_{T_m'}(y_1)  \setminus \L_{T_m'}(z)| +  |\L_{T_m'}(y_2) \setminus \L_{T_m'}(y_p)| \leq 2(d+d')$. 
     Otherwise, $|(\L_{T_m'}(y)  \setminus \L_{T_m'}(z)) \setminus \bigcup_{y' \in Z_2} \L_{T_m'}(y'))| \leq |\L_{T_m'}(y_1)  \setminus \L_{T_m'}(z)| +  |\L_{T_m'}(y_2) \setminus \L_{T_m'}(y_p)| +  |\L_{T_m'}(y_s)| \leq 3(d+d')$. 
     
     We have now shown that $|Z_2| \leq 2$ and  $|(\L_{T_m'}(y)  \setminus \L_{T_m'}(z)) \setminus \bigcup_{y' \in Z_2} \L_{T_m'}(y'))| \leq 4(d+d')$.
     The final property of $Z_2$ we wish to show is that for any $y' \in Z_2$ and any $T' \in P^*$, we have that $x \notin \L_{T'}(y'_{T'})$.
     
     So suppose  $x \in \L_{T'}(y'_{T'})$.
       Let $\hat{y}_1, \dots, \hat{y}_{d+d'}$ be $d+d'$ labels in $\L_{T_m'}(y_2)\setminus L_{T_m'}(y')$. 
       Let $w_1, \dots, w_{d+d'}$ be $d+d'$ labels in $\L_{T_m'}(y')$.
       Then for each $i \in [d+d']$, $T_m'$ (and therefore $T_2$) contains the triplet $(\hat{y}_iw_i|x)$, but $T'$ contains the triplet $(xw_i|\hat{y}_i)$.
      Therefore  $confset(T',T_2)$ contains $d+d'$ sets whose only common element is $x$. 
      As this implies a contradiction, we must have   $x \notin \L_{T'}(y'_{T'})$.
    
    \medskip
    
    Now that $Z_1$ and $Z_2$ have been constructed, let $Z = Z_1\cup Z_2$. Note that $|Z|\leq 4$.
    Algorithm~\ref{alg:locationRestrictions} describes the construction of $y$ and $Z$ formally (see Figure~\ref{fig:anchors}).

    \begin{figure*}[h]
\centering
\includegraphics[width=0.5\textwidth]
{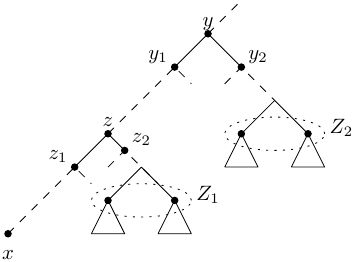}
\caption{Construction of $y$ and $Z = Z_1  \cup Z_2$, see Algorithm~\ref{alg:locationRestrictions}. Dashed edges represent parts of the tree that were omitted. Triangles represent parts of the tree that may contain more than $d+d'$ leaves.
}
\label{fig:anchors}
\end{figure*}
    
    We have shown above that for any $T' \in P^*$, $x$ is descended from $y_{T'}$ in $T'$ and not from $z'_{T'}$ for any $z' \in Z$, and so $x \in \L_{T'}(y_{T'}) \setminus \bigcup_{z'\in Z}\L_{T'}(z'_{T'})$.    
    As $|\L_{T_m'}(z) \setminus \bigcup_{z' \in Z_1} \L_{T_m'}(z'))| \leq 4(d+d')$ and  $|(\L_{T_m'}(y)  \setminus \L_{T_m'}(z)) \setminus \bigcup_{y' \in Z_2} \L_{T_m'}(y'))| \leq 4(d+d')$, we have 
    $|\L_{T_m'}(y) \setminus \bigcup_{z' \in Z} \L_{T_m'}(z'))| \leq 8(d+d')$.
    
    To analyze the complexity, note that we can calculate the value of $|\L_{T_m'}(u)|$  for all $u$ in $O(n)$ time using a depth-first search approach, together with the fact that $|\L_{T_m'}(u)| = |\L_{T_m'}(u_1)| + |\L_{T_m'}(u_2)|$ for any node $u$ with children $u_1, u_2$.
Then we can find $z$ in $O(n)$ time, and once we have found $z$  we can find $y$, and thence $z_1,z_2,y_1,y_2$, in $O(n)$ time. 
Similarly, once these nodes are found we can find the members of $Z$ in $O(n)$ time.
\qed
\end{proof}
    

%

\medskip

Using the claim, we may now construct a set $P'$ of $O\leq 16(d+d')+8$ trees on $\lblset \setminus (X_m\setminus \{x\})$, such for any $T' \in P^*$, $P'$ contains the tree $T'' = T' - (X_m\setminus \{x\})$.
Indeed, let $F$ be the set of arcs $uv$ in $T_m$ that exist on a path from $y$ to a node in $(\L_{T_m}(y) \setminus \bigcup_{z \in Z}\L_{T_m}(z')) \cup Z$.  
As 
$|\L_{T_m}(y) \setminus \bigcup_{z'\in Z}\L_{T_m}(z')| \leq 8(d+d')$, 
$|Z| \leq 4$
and $T_m$ is a binary tree, we have $|F| \leq 16(d+d')+8$.
For each $e \in F$, let $T_e$ be the tree obtained from $T_m$ by grafting $x$ onto the arc $e$ .
Let $P' = \{T_e: e \in F\}$.

Let $T'$ be a tree in $P^*$ and consider $T'' = T' - (X_m\setminus \{x\})$. 
Note that $T'' - \{x\} = T' - X_m = T_m$.
Therefore $u_{T''}$ is well-defined for every node $u \in V(T_m)$, and every node in $T''$ is equal to $u_{T''}$ for some $u\in V(T_m)$, except for $x$ and its parent in $T''$.
So let $w$ be the parent of $x$ in $T''$, $u_{T''}$ the parent of $w$ in $T''$, and $v_{T''}$ the child of $w$ in $T''$ that is not $x$.
Observe that $T''$ can be obtained from $T_m$ by grafting $x$ onto the arc $uv$.
Then it is enough to show that $uv \in F$.

To see that $uv \in F$,
first note that for each $z' \in \{y\} \cup Z$, $z'_{T''}$ is well-defined and $ (z'_{T''})_{T'} = z'_{T'}$ (see Lemma~\ref{lem:corrTransitive}).
Then as $x$ is descended from $(y_{T''})_{T'} = y_{T'}$ in $T'$, $x$ is descended from $y_{T''}$ in $T''$ (Lemma~\ref{lem:corrAncestor}).
Similarly, as $x$ is not descended from $(z'_{T''})_{T'} = z'_{T'}$ in $T'$ for any $z' \in  Z$, $x$ is not descended from $z_{T''}$ in $T''$.
Thus $x \in \L_{T''}(y_{T''}) \setminus \bigcup_{z' \in Z} \L_{T''}(z'_{T''})$.
It follows that $u_{T''}$ is a descendant of $y_{T''}$ in $T''$ (note that $y_{T''} \neq w$, as $w$ is  not the least common ancestor of any set of labels in $\lblset(T_m)$).
Also, $v_{T''}$ is not a descendant of  $z_{T''}$ for any $z' \in Z$, unless $v_{T''} \in \bigcup_{z\in Z} z_{T''}$, as otherwise $x$ would be a descendant of  such a $z_{T''}$.
Thus, $v_{T''}$ is either a member or an ancestor of  $(\L_{T''}(y_{T''}) \setminus \bigcup_{z' \in Z}\L_{T''}(z'_{T''}))) \cup \bigcup_{z' \in Z}z'_{T''}$.
It follows using Lemma~\ref{lem:corrAncestor} that $u$ is a descendant in $T_m$ of $y$, and $v$ is an ancestor of   $(\L_{T_m}(y) \setminus \bigcup_{z' \in Z}\L_{T_m}(z'))) \cup \bigcup_{z' \in Z}z'$.
Then $uv \in F$, as required.

\medskip

Now  that we have constructed our set $P'$, it remains to  find, for each $T_e \in P'$, every tree $T'$ on $\lblset$ such that $T' - (X_m\setminus \{x\}) = T_e$ and $T' - \{x\} = T_1 - \{x\}$. This will give us our set $P$, as for every $T' \in P^*$, $T' - (X_m\setminus \{x\})$ is a tree $T_e$ in $P'$, and  $T' - \{x\} = T_1 - \{x\}$.

Let $e = uv$, where $u,v \in V(T_m)$, and let $T_{1e}$ be the subtree of  $T_1 - \{x\}$  whose root is $v$, and has as its label set $v$ together with all labels in $X_m\setminus \{x\}$ descended from $u$ but not $v$. Then we have to try every way of adding $x$ into this tree. If $T_{1e}$ contains $t$ labels from $X_m$, then there are $2t-1$ places to try adding $x$.
Therefore $P$ will have at most $2|X_m| \leq 2(d+d')$ additional trees compared to $P'$, and so $|P| \leq 18(d+d')+8$.
Algorithm~\ref{alg:candidateTrees} gives the full procedure to construct $P$.

\begin{algorithm}
\caption{FPT algorithm to find candidate trees for $(T_1, T_i, x)$}
\begin{algorithmic}[1]
\Procedure{candidate-trees}{$T_1, T_2, x, d, d'$}
\Statex $T_1, T_2$ are two trees, $x$ is a label that cannot be moved in $T_2$ (but must be moved in $T_1$), $d$ is the maximum number of leaves we can remove in a tree, $d'$ is the maximum number of leaves we can move in $T_1$. 
For this pseudocode, every tree $T$ has a degree-$1$ root $r(T)$.
\State Find $X_m'$  such that $|X_m'| \leq d' + d - 1$ and $(T_1 - \{x\}) - X_m'  = (T_2 - \{x\}) - X_m'$  
\State Set $X_m = X_m' \cup \{x\}$
\State Set $T_m = T_1 - X_m$
\State Set $(y,Z) = \textsc{location-restriction}(T_m, T_2, X_m, x, d, d')$ 
 \Comment{$y,Z$ are nodes in $T_m$ such that roughly speaking, we may assume $x$ must become a descendant of $y$ but not of any $z' \in Z$.}
\State Set $U = \{u\in V(T_m): u \in Z$ or $u$ is a leaf descended from $y$ but not from any $z'\in Z\}$
\State Set $F = \{uv \in E(T_m): uv$ is on a path from $y$ to $U$\} 
\Comment{$F$ is the set of  edges we could graft $x$ onto.}
\State Set $P = \emptyset$ \Comment{Given $F$ we now begin constructing $P$.}
\State Set $T_1' = T_1 - \{x\}$
\For{$e = uv \in F$ with $u$ the parent of $v$} \Comment{Try grafting $x$ on $e$}
   \State Set $u_{T_1'} = $ the node in $T_1'$ corresponding to $u$
   \State Set $v_{T_1'} = $ the node in $T_1'$ corresponding to $v$
   \State Set $X_e= $ set of labels $l$ in $X_m\setminus \{x\}$ for which $l$ has an ancestor $v'$ in $T_1'$ with  $v'$ descended from $u_{T_1'}$, $v_{T_1'}$ descended from $v'$  
   \State \Comment{$X_e$ is the set of leaves of $T_1$ for which we have to subdivide $e$.}
   \State Set $U = v_{T_1'} \cup X_e$ 
   \State Set $E_e = \{u'v' \in E(T_1'): u'v'$ is on a path from $u_{T_1'}$ to $U$\} 
   \For{$u'v' \in E_e$}
      \State Constuct $T'$ from $T_1'$ by grafting $x$ on $u'v'$ 
      \State Set $P = P \cup  \{T'\}$
   \EndFor
\EndFor
\State Return $P$
\EndProcedure
\end{algorithmic}\label{alg:candidateTrees}
\end{algorithm}

To analyze the complexity, recall that we find $X_m$, and therefore construct $T_m$ and $T_m'$, in $O(n \log n)$ time. 
As shown above, we can find the node $y$ and set $Z$ in $O(n)$ time.
Given $y$ and $Z$, the set of arcs $F$ can be found in $O(n)$ time using a depth-first search approach. For each $e \in F$ it takes $O(n)$ time to construct $T_e$, and so the construction of $P'$ takes $O(|F|n) = O((16(d+d')+8)n)$ time.
Finally, the construction of of $P$ from $P'$ takes $O(|P|n) = O((18(d+d')+8)n)$ time.
Putting it all together, we have that the construction of $P$ takes $O(n(\log n + 18(d+d')+8))$ time.
\qed
\end{proof}

We will call the set of trees $P$ described in Lemma~\ref{lem:dont-check-too-many-trees}
the set of \emph{candidate trees} for $(T_1, T_2, x)$.

We are finally ready to give the proof of Theorem~\ref{thm:fpt-in-d}

\medskip
\noindent
 \textbf{Theorem~\ref{thm:fpt-in-d}}
 \emph{
 (restated).
\mastrld{} can be solved in time $O(c^d d^{3d}(n^3 + tn \log n))$, where  $c$ is a constant not depending on $d$ or $n$.
}

\begin{proof}

The outline for our algorithm is as follows. 
We employ a branch-and-bound algorithm, in which at each step we attempt to modify the input tree $T_1$ to become close to a solution.
We keep track of an integer $d'$, representing the maximum length of an LPR sequence between $T_1$ and a solution. Initially set $d' = d$.
At each step, if $d_{LR}(T_1,T_i) \leq d$ for each $T_i \in \T$ then $T_1$ is a solutioon, and we are done.
Otherwise, there must exist sime $T_i$ for which $d_{LR}(T_1,T_i)\geq d+d'$.
In this case, we calculate the $(d+d')$ disagreement kernel $S$  between $T_1$ and $T_i$ (using the procedure of Lemma~\ref{lem:disagreementKernel2}), and for each $x \in S$, attempt to construct a set $P$ of trees as in Lemma~\ref{lem:dont-check-too-many-trees}. For each $T' \in P$, we try replacing $T_1$ with $T'$, reducing $d'$ by $1$, and repeating the procedure. 
Algorithm~\ref{alg:fpt-d} describes the full procedure formally.

We claim that Algorithm~\ref{alg:fpt-d} is a correct algorithm for \mastrld, and runs in time $O(c^d d^{3d}(n^2 + tn \log n))$, for some constant $c$ not depending on $n$ or $d$.

First notice that if, in a leaf node of the branch tree created by Algorithm~\ref{alg:fpt-d}, a tree $T^*$ is returned, this occurs at line~\ref{line:returnsol} in which case it has been verified 
that $T^*$ is indeed a solution.
As an internal node of the branch tree returns a tree 
if and only if a child recursive call also returns a tree (the for loop on line~\ref{line:for-loop-P}), this shows that when the algorithm outputs a tree $T^*$, it is indeed a solution.

We next show that if a solution exists, then Algorithm~\ref{alg:fpt-d} will return one.
Suppose that $\T$ admits a solution, and let $T^*$ be a solution that minimizes $d_1 := d_{LR}(T_1, T^*)$,
with $d_1 \leq d'$.
We show that one leaf of the branch tree created by the algorithm returns $T^*$ (and thus the root of the branch tree also returns a solution, albeit not necessarily $T^*$).
This is done by proving that in one of the recursive calls made to 
{\sc mastrl-distance} on line~\ref{line:reccall}, 
the tree $T'$ obtained from $T_1$ satisfies $d_{LR}(T', T^*) = d_1 - 1$.
By applying this argument inductively, this shows that the algorithm will find $T^*$ 
at some node of depth $d_1$ in the branch tree of the algorithm.


First notice that since $d_{LR}$ is a metric, for each $T_i \in \T$,
$d_{LR}(T_1, T_i) \leq d_{LR}(T_1, T^*) + d_{LR}(T^*, T_i) \leq d' + d$, and so 
the algorithm will not return $FALSE$ on line~\ref{line:too-high-dist}.

If $T_1$ isn't a solution, then there is a tree of $\T$, say $T_2$ w.l.o.g., such that
$d_{LR}(T_1, T_2) > d$.
Notice that in this case, all the conditions of Lemma~\ref{lem:must-move-x} are satisfied, \textit{i.e.} 
$d_{LR}(T_1, T_2) > d$, and 
there are sets $X_1, X_2 \subseteq \lblset$ both of size at most $d$ such that 
$T_1 - X_1 = T^* - X_1$, $T_2 - X_2 = T^* - X_2$.  
Thus there is a minimal disagreement $X$ between $T_1$ and $T_2$, $|X| \leq d' + d$, 
and $x \in X$ such that $x \in X_1 \setminus X_2$.
By Lemma~\ref{lem:lpr-seq-equiv}, there is an LPR sequence $L = (x_1, \ldots, x_k)$ 
turning $T_1$ into $T^*$, where $\{x_1, \ldots, x_k\} = X_1$.
As $x \in X_1$, by Lemma~\ref{lem:lpr-order}, the leaves appearing in $L$ can be reordered, 
and we may assume that $x = x_1$.
Finally by Lemma~\ref{lem:dont-check-too-many-trees}, if $T'$ satisfies $d_{LR}(T', T^*) \leq d_1 - 1$ and $T'$ can be obtained from 
$T_1$ by an LPR move on $x$, then $T' \in P$.
As we are making one recursive call to {\sc mastrl-distance}for each tree in $P$, 
this proves that one such call replaces $T_1$ by $T'$ such that $d_{LR}(T', T^*) = d_1 - 1$.

As for the complexity,
recall from Lemma~\ref{lem:disagreementKernel2} that the  $(d + d')$-disagreement kernel $S$ computed in line 8 contains at most $8d^2$ labels.
Therefore
when Algorithm~\ref{alg:fpt-d} enters the 'for' loop of line 9, it branches into at most $8d^2$ cases, one for each $x \in S$.
Within each of these cases, the algorithm enters at most $|P|$ recursive calls, each of which decrements $d'$.
As $|P| \leq 18(d+d')+8 \leq 36d + 8$ by Lemma~\ref{lem:dont-check-too-many-trees},
a single call of the algorithm splits into at most  $8d^2(36d + 8) = O(d^3)$,
each of which decrements $d'$.
Therefore, the branching tree created by the algorithm has
degree at most $c d^3$ (for some constant $c$) and depth at most $d$, and so $O(c^d d^{3d})$ cases are considered.
 

As $d_{LR}(T_1, T_i)$ can be calculated in $O(n \log n)$ time for each $T_i$, a single call of lines 2-5 of the algorithm takes $O(tn\log n)$ time. A single call of lines 6-8 takes  $O(n^2)$ time by Lemma~\ref{lem:dont-check-too-many-trees}.
Thus the total time for all calls of lines 2-8 is $O(c^d d^{3d}n (n^2 + t \log n)$.
Each call of line 10 occurs just before a recursive call to the algorithm, as so line 10 is called at most 
$O(c^d d^{3d})$ times.
A single call of line 10 takes 
$O(n(\log n + 18(d + d') + 8)) = O(n(\log n + 36d)) $ time by Lemma~\ref{lem:dont-check-too-many-trees}, and so the total time for all calls of line 10 is
$O(c^d d^{3d}n(\log n + 36d))$.
Thus in total, we have that the running time of the algorithm is $O(c^d d^{3d} (n^2 + n (t \log n + 36d))$.
As we may assume $d\leq n$, this simplifies to $O(c^d d^{3d}(n^2 + tn \log n))$.
\qed
\end{proof}

\begin{algorithm}
\caption{FPT algorithm for parameter $d$.}\label{alg:fpt-d}
\begin{algorithmic}[1]
\Procedure{mastrl$-$distance}{$\T = (T_1, T_2, \ldots, T_t), d, d'$}
\Statex $\T$ is the set of input trees (represented as a sequence to distinguish $T_1$ from the other trees), $d$ is the maximum number of leaves we can remove in a tree, $d'$ is the maximum number of leaves we can move in $T_1$, which should be initially set to $d$.
\If{$d_{LR}(T_1, T_i) \leq d$ for each $T_i \in \T$} 
   \State Return $T_1$ \label{line:returnsol}
\ElsIf{there is some $T_i \in \T$ such that $d_{LR}(T_1, T_i) > d' + d$}
   \State Return FALSE   \#handles the $d' \leq 0$ case \label{line:too-high-dist}
\Else  \Comment{here we `guess' a leaf prune-and-regraft move on $T_1$}
   \State Choose $T_i \in \T$ such that $d_{LR}(T_1, T_i) > d$
  \State Set $S = \textsc{disagreement-kernel}(d+d',T_1, T_i)$  \label{line:dis-kernel}
   \For{$x \in S$}  \Comment{we are `guessing' that $x$ should go where $T_i$ wants it.} \label{line:for-loop-P}
      \State Set $P = \textsc{candidate-trees}(T_1,T_i,x,d,d')$      \label{line:compute-P}
      \State $T^* = FALSE$
      \For{$T \in P$}
           \State $T' = \textsc{mastrl$-$distance}((T, T_2, \ldots, T_t), d, d' - 1)$ \label{line:reccall}
           \State \algorithmicif~$T'$ is not $FALSE$, let $T^* := T'$
      \EndFor
      \State Return $T^*$
    \EndFor
\EndIf
\EndProcedure
\end{algorithmic}
\end{algorithm}

\end{document}